\documentclass[11pt]{article}

\usepackage{amsmath,amsthm,amsfonts,amscd,latexsym,amssymb,amsbsy,dsfont,mathrsfs,color}
\usepackage[small, centerlast, it]{caption}
\usepackage{epsfig}
\usepackage[titles]{tocloft}
\usepackage[latin1]{inputenc}
\usepackage{verbatim} 
\usepackage[active]{srcltx} 
\usepackage{fancyhdr}
\usepackage{caption}
\usepackage{enumerate}
\usepackage{color}
\usepackage{hyperref}

\definecolor{NiColor}{RGB}{77,77,255}
\definecolor{NiColoRed}{RGB}{255,77,77}
\definecolor{NiCitation}{RGB}{77,255,77}

\def\sp{\hskip -5pt}
\def\spa{\hskip -3pt}

\newsymbol\rest 1316    
\def\emptyset{\varnothing} 
\def\b1{{1\!\!1}}

\def\cB{\mathscr{B}}

\def\cD{{\ca D}}
\def\cF{{\ca F}}

\def\cL{\mathscr{L}}

\def\cN{{\ca N}}

\def\cS{\mathscr{S}}
\def\cE{\mathscr{E}}

\def\sK{{\mathsf K}}

\def\sH{{\mathsf H}}

\def\bC{{\mathbb C}}           

\def\bI{{\mathbb I}}

\def\bN{{\mathbb N}}

\def\bR{{\mathbb R}}

\def\gA{{\mathfrak A}}       
\def\gB{{\mathfrak B}}

\def\beq{\begin{eqnarray}}
\def\eeq{\end{eqnarray}}

\newcommand{\ca}[1]{{\cal #1}}         

\usepackage{sectsty}
\sectionfont{\normalsize}
\subsectionfont{\normalsize\normalfont\itshape}
\newtheoremstyle{TheoremStyle}
{3pt}
{3pt}
{\slshape}
{}
{\bf}
{:}
{.5em}
{}

\theoremstyle{TheoremStyle}
\newtheorem{theorem}{Theorem}
\newtheorem{corollary}[theorem]{Corollary}
\newtheorem{proposition}[theorem]{Proposition}
\newtheorem{lemma}[theorem]{Lemma}
\newtheorem{definition}[theorem]{Definition}
\newtheorem{remark}[theorem]{Remark}
\newtheorem{example}[theorem]{Example} 

\begin{document}


\par
\bigskip
\large
\noindent
{\bf The notion of observable and the moment problem for $*$-algebras and their GNS representations}
\bigskip
\par
\rm
\normalsize


\noindent  {\bf Nicol\`o Drago$^{a}$}, {\bf Valter Moretti$^{b}$}\\
\par

\noindent 
  Department of  Mathematics, University of Trento, and INFN-TIFPA \\
 via Sommarive 14, I-38123  Povo (Trento), Italy.\\
$^a$nicolo.drago@unitn.it,  $^b$valter.moretti@unitn.it\\

 \normalsize

\par

\rm\normalsize

\noindent {\small January, 29, 2020}

\rm\normalsize


\par

\begin{abstract}
We address some usually overlooked issues concerning the  use of $*$-algebras in quantum theory and their  physical interpretation.
If $\mathfrak{A}$ is a  $*$-algebra describing a quantum system and $\omega: \mathfrak{A}\to\mathbb{C}$ a state, 
we focus  in particular on  the interpretation of $\omega(a)$ as expectation value for an algebraic observable  $a=a^*\in\mathfrak{A}$, studying 
the problem of finding a probability measure reproducing the moments $\{\omega(a^n)\}_{n\in\mathbb{N}}$.
	This problem enjoys a close relation with the selfadjointness of the (in general only symmetric) operator $\pi_\omega(a)$ 
in the GNS representation of  $\omega$ and thus  it has important consequences for the interpretation
 of $a$ as an  observable.
	We provide physical examples (also from QFT) where the moment problem for $\{\omega(a^n)\}_{n\in\mathbb{N}}$ does 
not admit a unique solution.
	To reduce this ambiguity, we consider the moment problem for the sequences $\{\omega_b(a^n)\}_{n\in\mathbb{N}}$, 
being $b\in\mathfrak{A}$ and $\omega_b(\cdot):=\omega(b^*\cdot b)$.
	Letting $\mu_{\omega_b}^{(a)}$ be a solution of the moment problem for the sequence $\{\omega_b(a^n)\}_{n\in\mathbb{N}}$,
 we introduce a consistency relation on the family $\{\mu_{\omega_{b}}^{(a)}\}_{b\in\mathfrak{A}}$.
	We prove a 1-1 correspondence between consistent families $\{\mu_{\omega_{b}}^{(a)}\}_{b\in\mathfrak{A}}$ and positive 
operator-valued measures (POVM) associated with the symmetric operator $\pi_\omega(a)$. 
	In particular there exists a unique consistent family of $\{\mu_{\omega_{b}}^{(a)}\}_{b\in\mathfrak{A}}$ if and only
 if $\pi_\omega(a)$ is maximally symmetric.
	This result suggests that a better physical understanding of the notion of  observable for general $*$-algebras 
should be based on POVMs rather than projection-valued measure (PVM).
\end{abstract}
\tableofcontents

\section{Introduction}\label{Sec: Introduction}

Physical observables in {\em quantum theory} -- namely the physical properties of the system under investigation -- can be conveniently described in the algebraic framework by a suitable $C^*$-algebra $\mathfrak{A}$. (The standard Hilbert space formulation is included   referring to the von Neumann algebra $\mathfrak{A}=\mathfrak{B}(\mathsf{H})$.) There is however 
a technically less rigid formulation where  the algebra $\mathfrak{A}$ is  no longer  a $C^*$-algebra but it is a less tamed  {\em $*$-algebra}. This second approach seems to be strictly necessary in {\em quantum field theory} (QFT) both in constructive \cite{Streater-Wightman-00} and perturbative \cite{Haag} approaches, also in curved spacetime and especially in presence of self interaction.
The $C^*$-algebraic setting is not at disposal there, since  almost all relevant operators are no longer bounded in any reasonable $C^*$- norm (see however \cite{Buchholz-Fredenhagen-19} as a recent new viewpoint on the subject). In particular,  locally-covariant  renormalization procedures are performed within the $*$-algebraic formulation (see \cite{book} for a recent review).

When dealing with unital $*$-algebras  for describing quantum systems, the notion of (quantum) {\em observable} has a more delicate status than in the $C^*$-algebra formulation  which, in our view,  has not yet received sufficient attention in the literature.

First of all, contrarily to the standard requirements of quantum theory, it is generally false that every Hermitian  element $a=a^* \in \gA$ -- to which we shall refer to as an {\em algebraic} observable -- is represented by an (essentially) selfadjoint operator $\pi_\omega(a)$ -- henceforth called {\em quantum} observable -- in a given GNS representation induced by an algebraic state $\omega :\gA \to \bC$.
In general, $\pi_\omega(a)$ results to be symmetric admitting many or none selfadjoint extensions -- \textit{cf.} example \ref{Ex: algebraically self-adjoint observables which are not self-adjoint as operators}.

This problem is actually  entangled  with  another  issue concerning  the standard physical interpretation of $\omega(a)$ as {\em expectation value} of the algebraic observable $a$ in the state $\omega$.
This interpretation  is usually  considered folklore without a critical discussion as it should instead deserve.
To be operationally effective, that interpretation would need a {\em probability distribution} $\mu_{\omega}^{(a)}$ over $\bR$ -- which may be provided by the spectral measure of $\overline{\pi_\omega(a)}$ if the latter is selfadjoint, \textit{i.e.} when the {\em algebraic} observable $a$ defines a {\em quantum} observable in the GNS representation of an algebraic state.

Independently of  the existence of a spectral measure, the problem of finding $\mu_\omega^{(a)}$ can be tackled in the framework of the  more general {\em Hamburger moment problem}, looking for a probability measure whose moments coincide to the known values $\{\omega(a^n)\}_{n\in\mathbb{N}}$.
However, in the general case, there are many such measures for a given pair $(a,\omega)$, independently of the fact that  $\pi_\omega(a)$ admits one, many or none selfadjoint extensions -- \textit{cf.} example \ref{Ex: uniqueness of moment problem and self-adjointness of the operator are not related}.
Therefore a discussion on the possible physical meaning of these measures seems to be  necessary. 

We stress that all these problems are proper of the $*$-algebra approach whereas they are almost automatically solved when the structure is enriched to a $C^*$-algebra.
For instance the above-mentioned spectral measure always exists when $\gA$ is a $C^*$-algebra,   since $\pi_\omega(a)$ is always selfadjoint if $a=a^*$ in that case.
Stated differently, the notion of algebraic and quantum observables always agree  in the $C^*$-algebraic setting.
However as already said, $*$-algebras are more useful in real applications  especially in QFT and, in that sense, they are  more close to physics than $C^*$-algebras.

The goal of this work is to analyze the above mentioned issues also producing some examples and counterexampled arising from elementary formalism of {\em quantum mechanics} (QM) and QFT.
Even if  these problems cannot be avoided, we  present some partially positive results.

The technical  objects of our discussion will consist of a $*$-algebra $\gA$, an Hermitian element $a$, the class of deformed states 
$\gA \ni c \mapsto \omega_b(c) := \omega(b^* c b)/\omega(b^*b)$, with $b\in \gA$, constructed out of a given initial state $\omega$.
We also consider families of measures $\{\mu_{\omega_{b}}^{(a)}\}_{b\in\mathfrak{A}}$ over $\bR$ such that $\mu_{\omega_{b}}^{(a)}$ solve the corresponding Hamburger moment problem for the moments $\{\omega_b(a^n)\}_{n\in\mathbb{N}}$.

As a first result, we establish that if the measures $\mu_{\omega_b}^{(a)}$ are uniquely determined for every $b$ and a fixed pair $(a, \omega)$, then $\overline{\pi_\omega(a)}$ and every $\overline{\pi_{\omega_b}(a)}$ are selfadjoint.
This provides a sufficient condition for which an algebraic observable defines quantum observables referring to a certain class of algebraic states.
The converse statement is however false -- \textit{cf.} example \ref{Ex: uniqueness of moment problem and self-adjointness of the operator are not related} -- as it follows from elementary models in QM and QFT.

As a second, more elaborated, result we prove that for fixed $a^*=a\in \gA$ and $\omega$, when assuming some physically natural coherence constraints on the measures $\mu_{\omega_b}^{(a)}$ varying $b\in \gA$, the admitted  families of constrained  measures $\{\mu^{(a)}_{\omega_b}\}_{b\in \gA}$ are  one-to-one with all possible positive operator-valued measures (POVMs) associated to the symmetric operator $\pi_\omega(a)$ through  Naimark's decomposition procedure for symmetric operators.
We shall refer to each of these POVM as a {\em generalized} observable.
Therefore, $\overline{\pi_\omega(a)}$ is maximally symmetric if and only if  there is only one such measure for every fixed $b$ -- that is, there is a unique generalized observable associated with the given algebraic observable through the GNS representation of $\omega$.
These unique measures are those induced by the unique POVM of $\overline{\pi_\omega(a)}$, which is a projection-valued measure,  if the operator is selfadjoint.
\paragraph{Structure of this work.} 
This paper is organized as follows.
Section \ref{Sec: Issues in the interpretation of a as an observable} discuss in detail the issues arising in the interpretation of algebraic  observables  $a=a^*\in\mathfrak{A}$ of a $*$-algebra $\mathfrak{A}$ as quantum observables in their GNS representations.
Section \ref{Sec: Selfadjointness of pi(a) and uniqueness of moment problems} contains a first result concerning essential selfadjointness of 
$\pi_\omega(a)$ when the measures solving the moment problem for every
deformed state are unique.
Section \ref{Sec: The notion of POVM and its relation with symmetric operators} contains a recap of the basic theory of POVMs and the theory of generalized selfadjoint extensions of symmetric operators -- some complements appear also in  appendix \ref{Appendix: on POVM}.
Section \ref{Sec: Generalized observable pi(a) and expectation-value interpretation of omega(a)} is the core of the work where are established the main theorems arising from  the two issues discussed in the introduction.
Section \ref{Sec: Conclusions and open problems} offers a summary of the results established in the paper and some open issue.
Appendix \ref{Appendix: on POVM} includes some complements about reducing subspaces, generalized selfadjoint extensions of symmetric operators and present  the proofs of some technical propositions.

\paragraph{Notation and conventions.}
We adopt throughout the paper  the standard definition of {\bf complex measure} \cite{Rudin87}: a map
 $\mu :\Sigma \to \bC$ which is  {\em unconditionally $\sigma$-additive}  over the $\sigma$-algebra $\Sigma$.  With this definition 
the {\em total variation} $||\mu|| :=|\mu|(\Sigma)$ turns out to be  finite. \\
We adopt standard  notation and definitions and, barring the symbol of the adjoint operator and that of scalar product, they are the same as in \cite{Moretti2017}. In particular,
an operator in  a Hilbert space $\sH$ is indicated  by $A: D(A) \to \sH$, or simply $A$, where the domain  $D(A)$ is always supposed to be a linear subspace of $\sH$.
The scalar product $\langle x|y\rangle$ of a Hilbert space is supposed to be antilinear in the {\em left} entry.
$\gB(\sH)$  denotes the $C^*$-algebra of bounded operators $A$  in Hilbert space $\sH$ with $D(A)=\sH$. $\cL(\sH)$ indicates the  lattice of orthogonal projectors over the Hilbert space $\sH$. The adjoint of an operator $A$ in a Hilbert space is always denoted by $A^\dagger$, while the symbol $a^*$ indicates the adjoint of an element $a$ of a $*$-algebra. A representation of unital $*$-algebras is supposed to preserve the identity.

The closure of a closable operator $A: D(A) \to \sH$ is indicated by $\overline{A}$.
If $A:D(A) \to \sH$ and $B: D(B) \to \sH$ we assume the usual convention concerning standard domains, in particular:
(i) $D(BA) := \{\psi \in D(A) \:|\: A\psi \in D(B)\}$;
(ii) $D(A+B):= D(A)\cap D(B)$;
(iii) $D(aA):=D(A)$ for $a\in \bC \setminus \{0\}$ and $D(0A):= \sH$.
The symbol $A \subset B$ permits the case $A=B$. If $A$ and $B$ are operators $A \subset B$ means that $D(A) \subset D(B)$ and $B\sp\rest_{D(A)}=A$.
An operator $A$ in  a Hilbert space $\sH$ is said to be  {\bf Hermitian} if $\langle Ax|y\rangle = \langle x|Ay\rangle$ for every $x,y \in D(A)$. A Hermitian operator $A$  is {\bf symmetric} if $D(A)$ is dense in $\sH$ (equivalently, $A \subset A^\dagger$). A symmetric operator is {\bf selfadjoint} if $A=A^\dagger$, {\bf essentially selfadjoint}
if it admits a unique selfadjoint extension (equivalently, if $\overline{A}$ is selfadjoint), in this case $\overline{A}$ is the unique selfadjoint extension of $A$.
A {\bf conjugation} in the Hilbert space $\sH$ is an {\em antilinear} isometric map $C:\sH \to \sH$
such that $CC = I$, where $I$ always denotes the {\bf identity operator} $I :\sH \ni x \mapsto x \in \sH$.

\section{Issues in the interpretation of $a=a^*\in\mathfrak{A}$ as an observable}
\label{Sec: Issues in the interpretation of a as an observable}

This section has the twofold goal of presenting the problems discussed within this work and introducing part of the mathematical machinery used in the rest of the paper.
For the  generally used notation and conventions not directly explained in the text\footnote{When a  new term is introduced and defined, its name  appears in {\bf boldface} style.}, see Section \ref{Sec: Introduction}.

\subsection{$*$-algebras, states, GNS construction}\label{Sec: *-algebras, states, GNS construction}
There are at least two approaches to quantum theories. The most known is the Hilbert-space formulation where the physical observables of a quantum system are represented by {\bf quantum observables}, namely  (essentially) selfadjoint operators in the Hilbert space of the system. There,  (normal) states are given by positive trace-class unit-trace operators and pure states are unit vectors up to phases.

The second approach instead relies upon the structure of  unital {\bf $*$-algebra} $\gA$, i.e.,  an associative complex algebra equipped with an identity element $\bI$ and an antilinear  involution $^*$.
This  is the most elementary mathematical machinery to describe and handle the set of observables of a quantum system in the algebraic formalism.
{\bf  Algebraic observables} are here the  elements  $a \in \gA$  which are {\bf Hermitian} $a=a^*$. 
This approach is in particular suitable when dealing with the algebra of {\em quantum fields}  (see, e.g., \cite{KM,Wald,Kay-Wald-91,Haag,Streater-Wightman-00}).

\begin{example} \label{exQFT} $\null$\\
{\em  {\bf (1)} If $(M,g)$ is a given {\em globally hyperbolic spacetime} \cite{book},  
a complex unital $*$-algebra $\gA(M,g)$ is associated to a free {\em Klein-Gordon} real scalar field $\Phi$. It is
 uniquely defined \cite{KM} by requiring the elements of $\gA(M,g)$ are finite   complex linear combinations of the  
identity $\bI$ and finite products of elements $\Phi(f)$, called  (abstract) {\bf  quantum field operators}, where $f \in C_c^\infty(M;\bR)$
 are real-valued {\em smearing functions}
and the following requirements are true for $a,b \in \bR$, $f,h \in C_c^\infty(M;\bR)$.
\begin{itemize}
\item[(i)] ($\bR$-linearity) $\Phi(af+bh) = a\Phi(f)+b\Phi(h)$;

\item[(ii)] (Hermiticity) $\Phi(f)^* = \Phi(f)$;

\item[(iii)](field equations) $\Phi(Pf)=0$;

\item[(iv)](bosonic commutation relations) $[\Phi(f), \Phi(g)] = iE(f,g) \bI$.
\end{itemize}
Above $P: C^\infty(M) \to C^\infty(M)$ is the {\em Klein-Gordon operator} for a given squared mass $m^2\geq 0$, and referred to the metric $g$  and  $E \in \cD'(M\times M)$ is
 the {\em causal propagator} of $P$.  It turns out that $\Phi(f)= \Phi(f')$ if and only if $E(f-f')=0$ -- where $(Ef)(x):= E(x,f)$ \cite{Kay-Wald-91, Wald, KM} -- which, in turn, means 
$f-f' = Pg$ for some $g\in C_c^\infty(M)$.\\
Field operators $\Phi(f)$ are in particular {\em algebraic observables}.
When introducing a self-interaction -- different from the classical gravitational field, which is already encompassed in the 
formalism -- the $*$-algebra $\gA(M,g)$ has to be suitably enlarged in order to implement perturbative
 treatements of the dynamics and {\em renormalization}  procedures \cite{book}. \\
{\bf (2)}  Still assuming that $(M,g)$ is a (time-oriented) globally hyperbolic spacetime,  an alternative but 
equivalent construction, especially exploited  in Minkowski spacetime, is the {\bf symplectic formulation}, where the
quantum fields are viewed as formal field operators $\Phi[\varphi]$ smeared against real smooth {\em  solutions of the KG equation}  $\varphi$ with compact Cauchy data on every spacelike Cauchy surface $\Sigma$ of $(M,g)$.
The real linear space of these solutions will be denoted by $Sol[M,g]$.
The {\em algebraic observables} $\Phi[\varphi]$ are generators of a complex unital $*$-algebra $\gA[M,g]$ and are supposed to satisfy, for  $a,b \in \bR$ and $\varphi,\tilde{\varphi} \in Sol[M,g]$,
\begin{itemize}
\item[(i)] ($\bR$-linearity) $\Phi[a\varphi+b\tilde{\varphi} ] = a\Phi[\varphi]+b\Phi[\tilde{\varphi} ]$;
\item[(ii)] (Hermiticity) $\Phi[\varphi]^* = \Phi[\varphi]$;
\item[(iii)](bosonic commutation relations)  $[\Phi[\varphi], \Phi[\tilde{\varphi} ] ]= -i \sigma(\varphi, \tilde{\varphi}) \bI$,
\end{itemize}
where we have introduced the {\em symplectic form} on $Sol[M,g]$ \beq\sigma(\varphi, \tilde{\varphi}) :=\int_\Sigma (\varphi \nabla_n \tilde{\varphi}   - \tilde{\varphi}  \nabla_n \varphi) d\mu_\Sigma\:,\label{symp}\eeq
where  $\Sigma$ is any spacelike Cauchy surface of $(M,g)$, $\nabla_n$ denotes the derivative along  the future-oriented orthogonal direction with respect to $\Sigma$, and $\mu_\Sigma$ is the natural measure induced by $g$ on $\Sigma$.
The integral does not depend on the choice of $\Sigma$.
It turns out that $E : C_c^\infty(M) \to Sol[M,g]$ is surjective and, from that abd the other prperties of $E$, it is not difficult  to prove that \cite{Kay-Wald-91,Wald} there is a unique unital $*$-algebra isomorphism $\alpha : \gA(M,g) \to \gA[M,g]$ defined by
\begin{equation}\label{idsigmaPhi} \alpha(\Phi(h))= \Phi[Eh]\:, \quad \mbox{for every $h \in C_c^\infty(M)$.} \end{equation}
In this sense the formulation relying on $\gA(M,g)$ and that referring to $\gA[M,g]$ are equivalent.\\
If $(M,g)$ is the $(d+1)$-dimensional Minkowski spacetime, the symplectic formulation can be reformulated by enlarging  $Sol[M,g]$ to the space of
 real smooth solutions of KG equation with Cauchy data which are real $\bR^d$-Schwartz functions on every $x^0=$ constant $3$-space,  where $(x^0, x^1, \ldots, x^n) \in \bR^{d+1}$ deniting  any system of  Minkowskian coordinates  on $M$.}\hfill  $\blacksquare$
\end{example}

%

As is well known, an (algebraic) {\bf state} on a unital $*$-algebra $\gA$ is a linear map $\phi : \gA \to \bC$ which is {\bf positive} ($\phi(a^*a) \geq 0$ for $a \in \gA$),  and {\bf normalized} ($\phi(\bI) =1$).
Per definition a \textbf{non-normalized state} is a linear, positive functional $\omega\colon\mathfrak{A}\to\mathbb{C}$ such that $\omega(\mathbb{I})\neq 0$.
Notice that $\omega(\bI) = \omega(\bI\bI)= \omega(\bI^*\bI)>0$ therefore 
a non-normalized state defines a unique state  $\widehat{\omega}(a) := \omega(\mathbb{I})^{-1}\omega(a)$ 
for $a \in \gA$.
A state is said to be {\bf pure} if it is an extremal element of the convex body of states.
For an algebraic observable $a=a^*\in \gA$, the physical interpretation of $\phi(a)$  is the {\em expectation value} of  
$a$ in the state $\phi$. 
 
 In the following, we shall address the discussion about the possible interpretation of an algebraic observable $a\in \gA$ as a {\em quantum observable} in some  Hilbert space formulation. Later,  we  shall discuss the interpretation of $\phi(a)$ as an expectation value with respect to some probability measure and the interplay with the former issue. To this end, we list a few fundamental technical notions and results we shall exploit througout the work.

The basic link between the algebraic formalism and  the Hilbert space formulation of quantum theories 
 is provided by a celebrated construction developed by Gelfand, Naimark, and Segal and known as  {\em GNS construction} (see, e.g., \cite{KM,Moretti2019}). 
It   is valid for every algebraic state  over a unital $*$-algebra  and trivially extends  to non-normalized states.
The construction admits a more sophisticated topological version for $C^*$-algebras -- \textit{cf.} section \ref{Sec: Issue A: interpretation of pi(a) as observable}.

\begin{theorem}[GNS construction]\label{GNS}
Let  $\mathfrak{A}$ be a unital $*$-algebra and  $\omega : \gA \to \bC$ a (non-normalized) algebraic state. There exists a quadruple
 $(\mathsf{H}_\omega,  {\cal D}_\omega,\pi_\omega,\psi_\omega)$ 
called  {\bf GNS quadruple} of $(\gA, \omega)$,
where 
\begin{itemize}

\item[(1)] $\mathsf{H}_\omega$ is a Hilbert space whose scalar product is denoted by $\langle \: | \:\rangle$,

\item[(2)]  ${\cal D}_\omega\subset \mathsf{H}_\omega$ is a  dense subspace, 

\item[(3)]  $\pi_\omega: \mathfrak{A} \ni a \mapsto \pi_\omega(a)$ -- with $ \pi_\omega(a): {\cal D}_\omega \to {\cal D}_\omega$ -- is a unital algebra representation 
with the property that $\pi_\omega(a^*)\subset  \pi_\omega(a)^\dagger$, where $^\dagger$ denotes the adjoint in $\mathsf{H}_\omega$,

\item[(4)] $\psi_\omega \in \sH_\omega$ is a vector such that  

(i) ${\cal D}_\omega = \pi_\omega(\gA) \psi_\omega$,

(ii) $\omega(a)= \langle \psi_\omega| \pi_\omega(a) \psi_\omega\rangle$ for every $a\in \mathfrak{A}$,  in particular  $||\psi_\omega||^2 = \omega(\mathbb{I})$.
\end{itemize}
\noindent If $(\mathsf{H},  {\cal D},\pi,\psi)$ satisfies (1)-(4), then there is a surjective
 isometric map $U: \sH_\omega \to \sH$ such that $U({\cal D}_\omega)= {\cal D}$, $\pi(a) = U\pi_\omega(a) U^{-1}$ for $a\in \gA$, and $\psi = U\psi_\omega$.
\end{theorem}

   If $\gA$ is a unital $*$-algebra and $\omega$ a (non-normalized) state on it,  the set \beq \label{Gelfandideal}G_{(\mathfrak{A},\omega)} := \{a \in \mathfrak{A} \:|\: \omega(a^*a) =0\}\eeq  is a {\bf left-ideal} of  $\mathfrak{A}$ (a linear subspace such that $ba \in G_{(\mathfrak{A},\omega)}$ if $a \in G_{(\mathfrak{A},\omega)}$ and $b \in \mathfrak{A}$) as elementary consequence of Cauchy-Schwarz inequality 
and positivity of $\omega$.  $G_{(\mathfrak{A},\omega)}$ is called  {\bf Gelfand ideal}. 
For later convenience we recall that
the proof of the  GNS theorem 
leads to \beq\mathfrak{A}/G_{(\mathfrak{A},\omega)}= \mathcal{D}_\omega\:, \quad \pi_\omega(a)[b]= [ab] 
\quad \mbox{if  $a \in \gA\:,  [b] \in \mathfrak{A}/G_{(\mathfrak{A},\omega)}$}\:, \quad \mbox{and $\psi_\omega = [\bI]$.} \eeq
In particular, $\ker(\pi_\omega) \subset G_{(\mathfrak{A},\omega)}$ and 
  $\pi_\omega$ is faithful when $G_{(\mathfrak{A},\omega)} = \{0\}$\footnote{The converse does not hold,  since $\ker(\pi_\omega)$
 is a two-sided $*$-ideal and thus  $\ker(\pi_\omega) \subsetneq G_{(\mathfrak{A},\omega)}$ in the general case.}.  
%

\subsection{Issue A: interpretation of  $\pi_\omega(a)$  as a quantum observable}\label{Sec: Issue A: interpretation of pi(a) as observable}

When the unital $^*$-algebra $\mathfrak{A}$ is a {\bf unital $C^*$-algebra}  (i.e.,   a Banach space with respect to a norm satisfying $||ab|| \leq ||a||||b||$ and $||a^*a||=||a||^2$), then the GNS representation   $\pi_\omega(a)$ continuously extends to a $*$-algebra representation of $\gA$ to  $\gB(\sH_\omega)$ (see, e.g., \cite{Moretti2017}).
The extended representation denoted by the same symbol $\pi_\omega$ satisfies $\|\pi_\omega(a)\| \leq \|a\|$ if $a\in \gA$, where  equality holds for all $a$ if 
and only if $\pi_\omega$ is injective.
In particular, for $C^*$-algebras,  it holds $\pi_\omega(a^*)= \pi_\omega(a)^\dagger$ if $a \in \mathfrak{A}$.
Therefore, Hermitian elements of $\gA$  are  always  represented by (bounded) selfadjoint operators independently of $\omega$,  as it happens in the standard Hilbert space formulation of quantum theories.  Here, the two notions of observable always agree.

If $\mathfrak{A}$ is only a $*$-algebra, the picture becomes  more complex. 
Every operator $\pi_\omega(a)$ has the  common dense  invariant domain  ${\cal D}_\omega$ by definition and  it is  closable, since its adjoint operator 
$\pi_\omega(a)^\dagger$  extends $\pi_\omega(a^*)$ which has again the dense domain $ {\cal D}_\omega$. 
As a consequence,  $\pi_\omega(a)$ is a {\em symmetric operator} provided that  $a=a^*$.
If $\pi_\omega(a)$ is {\em essentially selfadjoint} for {\em every} $\omega$, then no issue pops out  and we can conclude that the {\em algebraic} observable $a$ has a definite meaning as a {\em quantum} observable also in the Hilbert space formalism.\footnote{
	A sufficient condition (by no means necessary!) assuring essential selfadjointeness of $\pi_\omega(a)$ for a fixed Hermitian element $a\in \gA$ and  {\em every} non-normalized state $\omega$ is that  \cite{Schmudgen90} there exist $b_\pm \in \gA$ such that $(a\pm i \bI)b_\pm = \bI$ (equivalently $b'_\pm (a\pm i \bI) = \bI$, where $b'_\pm =b_\mp^*$).
	This is because the written condition trivially implies that $\operatorname{Ran}(\pi_\omega(a) \pm iI) \supset \cD_\omega$ and thus $\operatorname{Ran}(\pi_\omega(a) \pm iI)$ is dense, so that the symmetric operator $\pi_\omega(a)$ is essentially selfadjoint (see, e.g. \cite[Thm. 5.18]{Moretti2017}).}. 
However,  we cannot {\em a priori} exclude the possibility that for some  (possibly unnormalized) state $\omega$ and some $a=a^* \in \gA$, $\pi_\omega(a)$  may admit  {\em different} selfadjoint extensions.
Or, worse,  that $\pi_\omega(a)$ admits  {\em no} selfadjoint extensions at all (its deficiency indices are different).
In these situations, an algebraic observable $a$ may fail to be  interpretable as a quantum observable, and therefore we are lead to the question whether $a$ should be though as a physical observablein any sense.
As shown by the following examples, there are concrete cases where $\pi_\omega(a)$ has either none or more than one self-adjoint extension.

\begin{example}\label{Ex: algebraically self-adjoint observables which are not self-adjoint as operators} $\null$\\
{\em {\bf (1)}   Let us define 
 the space $\cS$ of complex-valued smooth functions with domain  $[0,1]$ which vanish at $0$ and $1$ 
with all of their derivatives.
Consider the unital $*$-algebra $\gA$ of  differential operators acting on the function of the invariant space $\cS$, made of all finite complex linear combinations of finite compositions in arbitrary order of (i) the operator $P:= -i\frac{d}{dx}$ 
representing the {\em momentum} algebraic observable
for a particle confined in the box $[0,1]$, (ii)   {\em smoothed position} algebraic observables represented by  multiplicative operators $f\cdot$ induced by real-valued functions  $f \in \cS$, and (iii) the constantly $1$ function again acting multiplicatively and also defining the unit of the algebra.
The involution is $A^* := 
A^\dagger\spa\rest_{\cS}$ ($^\dagger$ being  the adjoint in $L^2([0,1], dx)\supset \cS$) so that $P^*=P$. 
We  stress that we are here considering $\gA$ as an abstract algebra (i.e., up to isomorphisms of unital $*$-algebras)  independently of the concrete realization in terms of operators we described above. Now consider the non-normalized state $\omega :\gA \to \bC$ defined as, where $dx$ is the Lebesgue measure on $\bR$,
$$\omega(A) := \int_0^1 \psi(x) (A\psi)(x) dx\quad A \in \gA\:.$$
Above, $\psi  \in \cS$ is a fixed non-negative function vanishing {\em only} at $0$ and $1$. The GNS structure is easy to be constructed taking advantage of the uniqueness part of  GNS theorem,
$$\sH_\omega = L^2([0,1], dx)\:, \quad \pi_\omega(A) =A\spa\rest_{\cD_\omega}\:, \quad \psi_\omega := \psi\:, $$
and $\cD_\omega$ is a suitable subspace of $\cS$ which however includes $C^\infty_c(0,1)$,  the space of smooth complex maps $f: [0,1] \to \bC$ supported in $(0,1)$, so that $\cD_\omega$ is dense in $L^2([0,1], dx)$ as  is due.  The deficiency spaces $N_\pm:=\ker(\pi_\omega(P)^*\mp i)=\operatorname{Ran}(\pi_\omega(P)\pm i)^\perp$ of $\pi_\omega(P)=P\spa\rest_{\cD_\omega}$
 are
\begin{align}
	N_\pm
	:=\left\{ g_\pm  \in L^2([0,1],dx)\:\left|\: \int_0^1 \overline{g_\pm} ( f'\mp f)dx =0\right. \quad \forall f \in \cD_\omega\right\} \nonumber \\
	=\left\{ g_\pm  \in L^2([0,1],dx)\:\left|\: \int_0^1 \overline{g_\pm(x) e^{\pm x}} \left( f(x) e^{\mp x}\right)'\right.dx =0
	\quad \forall f \in \cD_\omega\right\}\label{EED}
\end{align}
If $\cD_\omega$ were replaced by $C^\infty_c(0,1)$ in (\ref{EED}), \cite[Lemma 5.30]{Moretti2017} would 
imply $g_\pm(x)= c e^{\mp x}$ for $c\in \bC$.
However these functions would also satisfy (\ref{EED}) if $f \in \cS$, as one immediately proves per direct inspection.
Since $C^\infty_c(0,1) \subset \cD_\omega \subset \cS$, we conclude that 
$N_\pm =\mbox{span}\{e^{\mp x}\}$. Therefore the symmetric operator $\pi_\omega(P)$ is {\em not} essentially selfadjoint on its GNS domain $\cD_\omega$, but it admits  a one-parameter class of different selfadjoint extensions according to {\em von Neumann's extension theorem} (see, e.g. \cite[Thm. 5.37]{Moretti2017}).
Stated differently, $P$ is an algebraic observable which does not admit a  quantum-observable interpretation in the considered GNS Hilbert space.\\
{\bf (2)}
Let us define the space $\cE$ of complex-valued 
smooth functions with domain  $[0,+\infty)$ which vanish at $0$ with all of their derivatives and tend to $0$ with all of their derivatives for $x\to +\infty$ faster than every negative power of $x$.
Consider the unital $*$-algebra $\gB$ of  differential operators acting on the functions of the invariant space $\cE$, made of all finite complex linear combinations of finite compositions in arbitrary order of (i) the operator $P:= -i\frac{d}{dx}$ which again we would like to interpret  as the algebraic momentum observable
for a particle confined to stay in the half line, (ii)  smoothed position algebraic observables represented by multiplicative operators $f\cdot$ induced by real-valued  functions  $f \in \cE$, and (iii) the constantly $1$ function again acting multiplicatively and also defining the unit of the algebra. The involution is $A^* := 
A^\dagger\spa\rest_{\cE}$ ($^\dagger$ being  the adjoint in $L^2([0,+\infty), dx)\supset \cE$) so that $P^*=P$. 
As before,  we are here considering $\gB$ as an abstract algebra (i.e., up to isomorphisms of unital $*$-algebras)  independently of the concrete realization
we presented  above. Next consider the non-normalized state $\phi :\gB \to \bC$ defined as
\beq \phi(A) := \int_0^{+\infty} \chi(x) (A\chi)(x) dx\quad A \in \gB\:.\label{statephi}\eeq
Above, $\chi  \in \cE$ is a fixed non-negative function vanishing {\em only} at $0$. Uniqueness part of the GNS theorem proves that
the GNS structure is 
$$\sH_\phi = L^2([0,+\infty), dx)\:, \quad \pi_\phi(A) =A\spa\rest_{\cD_\phi }\:, \quad \psi_\phi  := \chi\:, $$
and $\cD_\phi$ is a suitable subspace of $\cE$ which however includes $C^\infty_c(0,+\infty)$,  the space of smooth complex maps $f: [0,+\infty) \to \bC$ whose  supports 
are included in $(0,+\infty)$.
 Hence $\cD_\phi$ is dense in $L^2([0,+\infty), dx)$ as is due.  The deficiency spaces $N_\pm$ of $\pi_\phi(P)=P\spa\rest_{\cD_\phi}$ can be computed easily
\begin{align}
	N_\pm &
	:=\left\{ g_\pm  \in L^2([0,+ \infty),dx)\:\left|\: \int_0^{+\infty} \overline{g_\pm} ( f'\mp f)dx =0\right.
	\quad \forall f \in \cD_\phi\right\}\nonumber \\
 	&= \left\{ g_\pm  \in L^2([0,+\infty),dx)\:\left|\: \int_0^{+\infty} \overline{g_\pm(x) e^{\pm x}} \left( f(x) e^{\mp x}\right)'\right.dx =0
 	\quad \forall f \in \cD_\phi\right\} \label{DDa}
 \end{align}
If $\cD_\phi$ were replaced by $C^\infty_c(0,+\infty)$ in (\ref{DDa}), \cite[Lemma 5.30]{Moretti2017} would 
imply $g_\pm(x)= c e^{\mp x}$ for $c\in \bC$.
Evidently $ce^x$ cannot be accepted as an element of $L^2([0,+\infty),dx)$ unless $c=0$, so that, since $C^\infty_c(0,+\infty)\subset \cD_\phi$,
we conclude that $N_-=\{0\}.$
The functions $ce^{-x}$ would fulfill (\ref{DDa}) even if $\cD_\phi$ were replaced by $\cE$, as one immediately proves per direct inspection.
Since $C^\infty_c(0,+\infty) \subset \cD_\phi \subset \cE$, we conclude that 
$N_+ =\mbox{span}\{e^{- x}\}$ whereas $N_- = \{0\}$.
Therefore the symmetric operator $\pi_\phi(P)$ is {\em not} essentially selfadjoint on its GNS domain $\cD_\phi$ and it does {\em not}  admit  selfadjoint extensions.
More strongly, as $N_+=\{0\}$ but $N_-\neq \{0\}$, we have that $\overline{\pi_\phi(P)}$ is {\em maximally symmetric}: it does not admit proper symmetric extensions \cite[Thm. 3, p.97]{Akniezer-Glazman-93}.}
Once again, the algebraic observable $P$  does not admit a  quantum-observable interpretation in the considered GNS Hilbert space.
\hfill $\blacksquare$\\
\end{example}

In summary,  it seems that, dealing with $*$-algebras which are not $C^*$-algebras, there is not a perfect match between the notion of algebraic observable (Hermitian element of $\gA$) and that of quantum observable (selfadjoint operator in the (GNS) Hilbert space formulation).
In particular, Hermitian elements of $*$-algebras are usually represented by merely symmetric operators in the GNS representations with many or none selfadjoint extensions.
This issue is relevant when dealing with the problem of measurement, see in particular  \cite{Fewster-Verch-2019,Fewster} where the problem  of measurement is discussed for the algebra of quantum fields  in a locally covariant framework.

\subsection{Issue B: interpretation of $\omega(a)$ as expectation value and the moment problem}
Let us pass to discuss  the interpretation of $\omega(a)$ as {\em expectation value} for $a=a^*\in \gA$,  where $\omega$ is a (possibly non-normalized) state\footnote{In the non-normalized case, the meaning  of expectation value would be actually  reserved to $\widehat{\omega}(a)=\omega(\bI)^{-1}{\omega}(a)$, though, for shortness,  we improperly also call $\omega(a)$ expectation value  in the rest of the work.}.
To rigorously accept this folk physical interpretation, we should assume that  the pair $(a,\omega)$ admits a physically meaningful uniquely associated  positive $\sigma$-additive  measure  $\mu_{\omega}^{(a)}$ over $\bR$ such that 
\beq \omega(a) = \int_\bR \lambda d\mu_{\omega}^{(a)}(\lambda)\:.\label{MC}\eeq
It is natural to also suppose  that $\mu$ is defined on the Borel $\sigma$-algebra $\cB(\bR)$, since this is the case for measures arising from the spectral  theory as it is standard in quantum theories.
Identity (\ref{MC}) is far from being  able  to determine $\mu_{\omega}^{(a)}$.
However, the structure of $*$-algebra permits us  to define real polynomials of algebraic observables and  $\omega$ does assign values  to all those algebraic observables. 

	In particular, for $a=a^*\in\mathfrak{A}$ and $n\in\mathbb{N}$, it is natural to interpret $a^n$ as the algebraic observable whose values are $\lambda^n$ if $\lambda$ is a value attained by $a$.
	For this reason we shall strengthen (\ref{MC}) by requiring that
	\beq \omega(a^n) = \int_\bR \lambda^n d\mu_{\omega}^{(a)}(\lambda)\quad \mbox{for every $n \in \bN$}\:. \label{MP}\eeq

In this way, $\omega(a^n)$ is interpreted as the $n$-th {\em moment} of the unknown measure $\mu_\omega^{(a)}$.
Finding a finite  {\em positive Radon measure}  over $\bR$ when its moments are fixed is a quite  famous problem named {\em Hamburger moment problem},  extensively treated in the pure mathematical literature (see \cite{Schmudgen17} for a modern textbook on the subject). 

\begin{remark}\label{remarkRB}
{\em A {\bf positive Radon measure} is a positive $\sigma$-additive measure defined on the Borel sets of a Hausdorff locally-compact space (here $\mathbb R$ equipped with the Euclidean topology) 
which is both outer and inner regular and assigns a finite value to every compact set.
All measures considered above  are necessarily finite  because $\omega(a^0) = \omega(\mathbb{I})$ exists in $[0,+\infty)$ by hypothesis.
 In $\mathbb  R^n$,  all  finite positive $\sigma$-additive  Borel measures  
are automatically Radon  in view of  \cite[Thm. 2.18]{Rudin87}.  Therefore, ``positive Radon measure''  can be equivalently replaced by ``positive $\sigma$-additive  Borel measure'' in the rest of the discussion related to the moment problem.}\hfill $\blacksquare$\\
\end{remark}

\noindent At this juncture,   for a given Hermitian  $a \in \gA$ and a given non-normalized state $\omega : \gA \to \bC$ ,  we should  tackle  two problems if we want to insist with the standard interpretation of $\omega(a)$ as expectation value.  
\begin{itemize}
\item[(M1)] Does a positive $\sigma$-additive  Borel measure  $\mu^{(a)}_\omega$ over $\bR$ satisfying  (\ref{MP}) exist? 
\item[(M2)] Is it unique?
\end{itemize}

Issues A and B are related in several ways. Here is a first example of that interplay arising when facing (M1) and (M2).  If $\pi_\omega(a)$ is essentially selfadjoint, a measure as in (M1) directly arises form the GNS construction. It is simply constructed out of the projection-valued measure (PVM) $P^{(\overline{\pi_\omega(a)})} : \cB(\bR) \to \gB(\sH_\omega)$ of the selfadjoint operator $\overline{\pi_\omega(a)}$
over $\sH_\omega$:
\beq\label{muPVM}
\mu_\omega^{(a)}(E) := \langle \psi_\omega| P^{(\overline{\pi_\omega(a)})}(E) \psi_\omega \rangle\:, \quad E\in \cB(\bR)\:. 
\eeq
This opportunity is always present if $\gA$ is a unital $C^*$-algebra, since $\pi_\omega(a) \in \gB(\sH_\omega)$ is selfadjoint in that case.  Concerning (M2), it is possible to prove that the measure defined in (\ref{muPVM}) is also {\em unique} when $\gA$ is a $C^*$-algebra.

\begin{proposition}\label{P1} Let $\gA$ be a unital $*$-algebra, $a=a^* \in \gA$, and $\omega : \gA \to \bC$ is a non-normalized state. The following facts hold.
\begin{itemize}
\item[{\bf (a)}] If $\pi_\omega(a)$ is essentially selfadjoint,   then there exists
 a (necessarily finite) positive $\sigma$-additive  Borel measure $\mu^{(a)}_\omega : \cB(\bR) \to [0,+\infty)$
satisfying (\ref{MP}).  

\item[{\bf (b)}] If furthermore $\gA$ is a $C^*$-algebra, then the measure $\mu^{(a)}_\omega$ is unique.
\end{itemize}
\end{proposition}

\begin{proof} (a) The measure  (\ref{muPVM}) 
is a finite  positive $\sigma$-additive  Borel measure over $\bR$ due to standard properties of spectral measures, it also 
satisfies  (\ref{MP}). Indeed, since $D([\overline{\pi_\omega(a)}]^n) \supset D(\pi_\omega(a)^n) \supset D(\pi_\omega(a^n)) = \cD_\omega \ni \psi_\omega$, from spectral theory (see. e.g. \cite{Moretti2017}) we have
$$\int_\bR \lambda^n d\mu^{(a)}_\omega =  \langle \psi_\omega| [\overline{\pi_\omega(a)}]^n \psi_\omega\rangle =  \langle \psi_\omega| \pi_\omega(a^n) \psi_\omega\rangle = \omega(a^n)\:.$$
(b) Let us suppose that $\gA$ is also a $C^*$-algebra. Since $|\omega(a^n)|\leq \omega(\bI) \|a\|^n$,  {\em Carleman's condition} \cite[Corollary 4.10]{Schmudgen17} assures that there exists at most one positive Radon measure satisfying (\ref{MP}). Observe that 
$\mu^{(a)}_\omega$ is a positive Radon measure in view of Remark \ref{remarkRB}.
\end{proof}
There are cases of unital $*$-algebras $\gA$ and (possibly unnormalized) states $\omega : \gA \to \bC$  such that  Hermitian elements $a\in \gA$ exist whose associated GNS operator $\pi_\omega(a)$ is not essentially selfadjoint -- see Example \ref{Ex: algebraically self-adjoint observables which are not self-adjoint as operators}. In this situation 
Proposition \ref{P1}  cannot be directly exploited. If $\pi_\omega(a)$ admits selfadjoint extensions (it is sufficient that it commutes with a {\em conjugation}) each of these selfadjoint extensions induces a measure $\mu_\omega^{(a)}$ as above.
However,  measures satisfying (M1) for the pair $(a,\omega)$ do exist,  and they are not necessarily unique, even if $\pi_\omega(a)$ does {\em not} admit any selfadjoint extension -- \textit{cf.} Examples \ref{Ex: algebraically self-adjoint observables which are not self-adjoint as operators}-\ref{Ex: uniqueness of moment problem and self-adjointness of the operator are not related} -- making the situation even more intricate.

The following proposition can be proved noticing that positivity of $\omega$ and its linearity implies that the set of candidate moments $m_n := \omega(a^n)$ satisfies the hypotheses of \cite[Thm. X.4]{RS2}.
However, we intend to provide a direct construction (which is nothing but the proof of the quoted theorem written with a different language).

\begin{proposition}\label{P2}
Let $\gA$ be a unital $*$-algebra, $a=a^* \in \gA$, and $\omega : \gA \to \bC$  a non-normalized state. Then there exists a  positive $\sigma$-additive  Borel measure $\mu^{(a)}_\omega : \cB(\bR) \to [0,+\infty)$
satisfying (\ref{MP}).  
\end{proposition}

\begin{proof}
	Define a subspace $\cD_\omega^{(a)}$ of $\cD_\omega$ as follows 
$\cD_\omega^{(a)} := \{ \pi_\omega(p(a))\psi_\omega \:|\: p: \bR \to \bC \:\: \mbox{polynomial}\}$ and define the closed subspace $\sH_\omega^{(a)}$  of $\sH_\omega$ as the closure of  $\cD_\omega^{(a)}$. By construction,  $\pi_\omega(a)$
leaves $\cD_\omega^{(a)}$ invariant and $\pi_\omega(a)\spa\rest_{\cD_\omega^{(a)}}$ is symmetric in $\sH_\omega^{(a)}$. Finally,  $\pi_\omega\spa\rest_{\cD_\omega^{(a)}}$ commutes with the {\em conjugation} $C : \sH_\omega^{(a)} \to \sH_\omega^{(a)}$ obtained as the unique 
continuous extension of the antilinear  isometric involutive map $\pi_\omega(p(a))\psi_\omega \mapsto \pi_\omega(p(a))^\dagger\psi_\omega$ over $\cD_\omega^{(a)}$ (use the fact that $a=a^*$). In view of the von Neumann criterion (see, e.g. \cite[Thm. 5.43]{Moretti2017}),  $\pi_\omega\spa\rest_{\cD_\omega^{(a)}}$ admits selfadjoint extensions $\widehat{a}_\omega$ on $\sH_\omega^{(a)}$. The same argument exploited to prove (a) in Proposition \ref{P1} restricted to $\sH_\omega^{(a)}$  concludes the proof because
\beq\label{muPVM2}
\mu_\omega^{(a)}(E) := \langle \psi_\omega| P^{(\widehat{a}_\omega)}(E) \psi_\omega \rangle\:, \quad E\in \cB(\bR)
\eeq
satisfies all requirements for every such selfadjoint extension $\widehat{a}_\omega$.
\end{proof}

In turn, the proof of the Proposition \ref{P2} raises another issue. Are all measures $\mu_{\omega}^{(a)}$ associated with a pair $(a,\omega)$ spectrally constructed from selfadjoint extensions over $\sH_\omega^{(a)}$ of $\pi_\omega(a)\sp\rest_{\cD_\omega^{(a)}}$ when this operator admits such extensions?
The answer is once again negative as it can be grasped  from the detailed discussion  about the moment problem in the operatorial approach appearing  in Ch.6 of \cite{Schmudgen17}. All  measures satisfying (M1)  are in fact spectrally  obtained by {\em enlarging} the Hilbert space $\sH_\omega^{(a)}$ {\em without reference to the original  common GNS Hilbert space} $\sH_\omega$.
 
This fact eventually  suggests that, in principle,  there could be a plethora of measures   associated with $(a,\omega)$ as solutions of the moment problem
 with dubious physical meaning, because they are vaguely related with the underpinning physical theory described by $\gA$ and $\omega$.
Our feeling is that   focusing on the {\em whole} class of the measures satisfying (M1) for a given pair  $(a,\omega)$ is probably a wrong approach to tackle the problem of the  interpretation of $\omega(a)$ as  expectation value.  Further physical meaningful information has to be added 
 in order to reduce the number of elements 
of the family of measures. 

There are some relevant papers in the literature on related topics, in particular  \cite{MDV1,MDV2} where a {\em non commutative} version of the moment 
problem (both existence and uniqueness) is 
addressed, carefully analyzed and solved for $*$-algebras, and finally also applied to the $*$-algebra of quantum fields. There,  the unknown is the non-normalized state on the $*$-algebra.
Conversely, the state is {\em a priori} known here and we focus on the standard moment problem referred to a measure associated to the state for fixed Hermitian element of the algebra.
Another important difference is that in \cite{MDV1,MDV2} natural separating $C^*$-seminorms on $\gA$ are exploited.
Here we instead stick to the minimal  purely algebraic structure of  $\gA$.  

\section{Selfadjointness of $\overline{\pi_\omega(a)}$ and uniqueness of moment problems}\label{Sec: Selfadjointness of pi(a) and uniqueness of moment problems}
In the discussion developed in section \ref{Sec: Issues in the interpretation of a as an observable}, when presenting the issues A and B, we completely overlooked the physically meaningful fact that other elements $b\in \gA$ than $a$ exist.
These elements can be used to generate new non-normalized states $\omega_b$ out of $\omega$ defined by $\omega_b(a) := \omega(b^*ab)$: $\omega_b$ is viewed as a deformation of $\omega$.
When we are given the triple 
$\gA, a, \omega$ (with $a=a^*$)  we also know the formal expectation values $\omega_b(a)$. We expect that these deformed  states and the associated measures $\mu_{\omega_b}^{(a)}$ solving the moment problem with respect to $\omega_b$ should enter the game. Theorem \ref{T1} below shows that it is the case.

\subsection{Deformations of a non-normalized state}

\begin{definition}\label{defpert} 
	 {\em If $\omega$ is a non-normalized state (or a state) over $\mathfrak{A}$, we will denote by  $\omega_b$ the non-normalized state, called $b$-{\bf deformation} of $\omega$ \begin{equation}
		\omega_b(a):=\omega(b^*ab) \quad \forall a \in \mathfrak{A}\:, \label{defomegab}
	\end{equation}
	where $b\in\mathfrak{A}$. 
	The limit case of the zero functional $\omega_b$ obtained from $b$ with $\omega(b^*b)=0$ is included, and we call that $\omega_b$ {\bf singular deformation}.}\hfill $\blacksquare$\\
\end{definition}
\noindent From the final uniqueness part of Theorem \ref{GNS}, the  GNS 
structure of a non-singular deformation $\omega_b$ is evidently
\begin{gather}\label{GNSb}
	(\mathsf{H}_{\omega_b}, {\cal D}_{\omega_b},\pi_{\omega_b},\psi_{\omega_b})=
	\left(\overline{\pi_\omega(\mathfrak{A})\psi_{\omega_b}}, \pi_\omega(\mathfrak{A})\psi_{\omega_b}, \pi_\omega\spa\rest_{{\cal D}_{\omega_b}},\pi_\omega(b)\psi_\omega\right)\,.
\end{gather}
If $\omega_b$ is singular we {\em define},
\beq (\mathsf{H}_{\omega_b}, {\cal D}_{\omega_b},\pi_{\omega_b},\psi_{\omega_b}) :=
	\left(\{0\}, \{0\}, 0, 0\right)\label{GNSb2}\:.\eeq

\begin{remark}\label{remarklabels} {\em Observe that $\omega_b= \omega_{b'}$ if  $[b]=[b']$ referring to  the Gelfand-ideal quotient. Therefore the non-normalized states $\omega_b$ would be  better labelled by the vectors in 
$\mathcal{D}_\omega = \mathfrak{A}/G_{(\mathfrak{A},\omega)}$.}\hfill $\blacksquare$
\end{remark}

\subsection{Selfadjointness of $\overline{\pi_\omega(a)}$ and uniqueness of moment problems for deformed  non-normalized states} The fact that focusing on the deformations $\omega_b$ goes towards the correct direction in order to clarify issues A and B  is evident from the following result, the first main result of the paper,  which connects  uniqueness of the measures $\mu_{\omega_b}^{(a)}$ with selfadjointness of $\overline{\pi_\omega(a)}$.

\begin{theorem}\label{T1}  Let $\gA$ be a unital $*$-algebra, $a^*=a \in \gA$, and $\omega : \gA \to \bC$ a non-normalized state.  Assume that the  
   finite  positive $\sigma$-additive  Borel measure $\mu^{(a)}_{\omega_b}
 : \cB(\bR) \to [0,+\infty)$
solving the moment problem  for every non-singular deformation $\omega_b$
\beq\label{bmoment}
\omega_b(a^n) = \int_\bR \lambda^n d\mu^{(a)}_{\omega_b}(\lambda) \quad \mbox{for every $n \in \bN$.}
\eeq
  is unique\footnote{If  $\omega_b$ is not singular, some $\mu^{(a)}_{\omega_b}$ exists due to Proposition \ref{P2}. If it is singular, the zero measure solves the moment problem.}.
 The  following facts hold for  every  deformation $\omega_b$.
\begin{itemize}
\item[{\bf (a)}] $\pi_\omega(a)$ is essentially selfadjoint in $\sH_\omega$ and, more generally,  all  operators $\pi_{\omega_b}(a)$ are essentially selfadjoint in the respective $\sH_{\omega_b}$. 

\item[{\bf (b)}] All measures $\mu^{(a)}_{\omega_b}$ are induced by the single PVM $ P^{(\overline{\pi_{\omega}(a)})}$ in the sense that
\begin{align}
&\mu_{\omega_b}^{(a)}(E) = \langle \psi_{\omega_b}| P^{(\overline{\pi_{\omega_b}(a)})}(E) \psi_{\omega_b} \rangle\:, \quad E\in \cB(\bR)\:,
\label{muPVMb}\\
&\label{dimenticata}
P^{(\overline{\pi_{\omega_b}(a)})}(E) = P^{(\overline{\pi_{\omega}(a)})}(E)\spa\rest_{\sH_{\omega_b}} \:, \quad E\in \cB(\bR)\:.
\end{align}
\end{itemize}
\end{theorem}

\begin{proof} (a) {\em Carleman's condition} \cite[Corollary 4.10]{Schmudgen17} assures that, if $\omega_b$  is singular, only the zero measure $\mu_{\omega_b}^{(a)}=0$ solves the moment problem.We can therefore assume that there is a unique measure solving the moment problem for every $b \in \gA$.
From \cite[Thm. 6.10]{Schmudgen17}  
translated into our GNS-like formulation as in the proof of Proposition \ref{P2} and using the notation introduced therein, 
we have that $\mu_{\omega_b}^{(a)}$ is unique if and only
if $\pi_{\omega_b}(a)\spa\rest_{\pi_{\omega_b}(\mathfrak{A}_a)\psi_b}$ is essentially selfadjoint in $\overline{\pi_{\omega_b}(\mathfrak{A}_a)\psi_b}$ -- where $\mathfrak{A}_a$ denotes the $*$-algebra generated by $a$.
Now observe that the set of all vectors
$\psi_{\omega_b} = \pi_\omega(b)\psi_\omega$, for $b \in \gA$, is dense in $\sH_\omega$ since it coincides with $\cD_\omega$, so that it is a dense set of uniqueness vectors for the symmetric 
 operator $\pi_\omega(a)$ in $\sH_\omega$, which is therefore essentially selfadjoint in view of  {\em Nussbaum's lemma}
  (Lemma on p. 201 of \cite{RS2}). Essential selfadjointness of $\pi_{\omega_b}(a)$
can be established similarly, taking  (\ref{GNSb}) and (\ref{GNSb2}) into account.  By hypothesis, fixing $b\in \gA$, also the measures $\mu_{\omega_{cb}}^{(a)}$ are uniquely determined by the deformations $\omega_{cb}$ with $c\in \gA$.  The set of all vectors $\psi_{\omega_{cb}} = \pi_\omega(cb)\psi_{\omega} = \pi_\omega(c) \psi_{\omega_b}$, for $c \in \gA$, is dense in $\sH_{\omega_b}$ since it coincides to $\cD_{\omega_b}$, so that it is a dense set uniqueness vectors for the symmetric 
 operator $\pi_{\omega_b}(a)$ in $\sH_{\omega_b}$ which is  essentially selfadjoint due to Nussbaum's lemma again.\\
 (b)  Assuming (\ref{dimenticata}), then (\ref{muPVMb})  is a trivial consequence of the uniqueness hypothesis and (a) of Proposition \ref{P1} applied to the non-normalized state $\omega_b$. 
Let us prove  (\ref{dimenticata}) to conclude. We know that $\overline{\pi_\omega(a)}$ admits $\cD_{\omega_b}$ as invariant subspace from (\ref{GNSb}) and
(\ref{GNSb2}), furthermore 
$\pi_\omega(a)\spa\rest_{\cD_{\omega_b}} = \pi_{\omega_b}(a)$ is essentially selfadjoint in the Hilbert space  $\sH_{\omega_b}$, which is a closed subspace of  $\sH_{\omega}$, and $\sH_{\omega_b}$ is the closure of 
$\cD_{\omega_b}$. Proposition \ref{propSR} 
implies that $\sH_{\omega_b}$ reduces $\overline{\pi_\omega(a)}$ (see Appendix \ref{reduction}) so that 
$\overline{\pi_{\omega_b}(a)}$ is the part of 
$\overline{\pi_\omega(a)}$ on $\sH_{\omega_b}$. Using the properties of the PVMs, it is easy to prove that if a closed subspace $\sH_0$ reduces a selfadjoint operator $T$, then the PVM of the part of $T$ on  $\sH_0$ (which is selfadjoint in view of Proposition \ref{priopevidente})  is the restriction to $\sH_0$ of the PVM of $T$.
Applying this result to $\overline{\pi_{\omega}(a)}$ and  $\overline{\pi_{\omega_b}(a)}$,
uniqueness of the PVM of a selfadjoint operator implies that the PVM of
$\overline{\pi_{\omega_b}(a)}$  is nothing but the restriction to $\sH_{\omega_b}$ of the PVM of $\overline{\pi_\omega(a)}$. This is just  (\ref{dimenticata}).
\end{proof}

\begin{example} 
	{\em A simple example of application of Theorem \ref{T1} is the following. (It is not physically interesting, but it just provides evidence that the very strong  hypotheses of Theorem \ref{T1} are fulfilled in some case.) Consider the unital Abelian $*$-algebra $\bC_{[0,1]}[x]$
made of all complex polynomials $p: [0,1] \to \bC$ in the variable $x$ with the involution defined as the standard point-wise complex conjugation, and the unit given by the constantly $1$ polynomial.
Define the state
 $$\omega(p) = \int_0^1 p(x) dx\quad p \in \bC_{[0,1]}[x]\:.$$
Using in particular the Stone-Weierstrass theorem and the uniqueness part of the GNS theorem, it is easy to prove that a GNS representation is $$\sH_\omega := L^2([0,1], dx)\:, \quad 
\pi_\omega(p) =p\cdot \;, \quad \cD_\omega = \bC_{[0,1]}[x]\:,\quad  \psi_\omega  =1$$ 
where $p\cdot$ denotes the polynomial $p$ acting as multiplicative operator on $\bC_{[0,1]}[x]$.
Evidently $\omega_q(p) = \int_0^1 |q(x)|^2 p(x)\: dx$.
Hence $|\omega_q( x^n)| \leq \frac{C_q}{n+1}$ where $C_q = \max_{[0,1]}|q|^2$.  {\em Carleman's condition} \cite[Corollary 4.10]{Schmudgen17} assures that there exists at most one positive Radon measure satisfying (\ref{bmoment}) for 
$b=p$ and $a=x$. Hence we can apply Theorem \ref{T1} and this means in particular that $\pi_\omega(x)$, i.e. the symmetric multiplicative operator $x$ with  domain consisting
 of the complex polynomials on $[0,1]$,  is essentially selfadjoint  in $L^2([0,1],dx)$.}\hfill $\blacksquare$\\
\end{example}

At this point, it may seem plausible that the statement (a)  of Proposition \ref{P1} can be reversed proving that, if $a^*=a\in \gA$, essential selfadjointness of all $\pi_{\omega_b}(a)$ is equivalent to uniqueness of all the measures $\mu_{\omega_b}^{(a)}$. Unfortunately life is not so easy as a consequence of the last item of the pair of  examples  below. 
The former example concerns 
again QM whereas the latter concerns elementary QFT as in Example \ref{exQFT}.

\begin{example}\label{Ex: uniqueness of moment problem and self-adjointness of the operator are not related} 
{\em {\bf (1)}
	Let $\mathfrak{A}_{\textrm{CCR},1}$ be the $*$-algebra generated by $I,Q,P$ defined as follows.
		$Q, P, I : \cS(\bR) \to \cS(\bR)$ are respectively the operators $$(Q\psi)(x) = x\psi(x)\;,\quad (P\psi)(x) = -i \frac{d}{dx}\psi(x)\:, \quad (I \psi)(x) = \psi(x) \:, \quad x \in \bR\:. $$
The finite linear combinations of finite compositions of these operators with arbitrary order  form a unital $*$-algebra with unit $I$, provided  the involution is defined as 
$A^* := A^\dagger\spa\rest_{\cS(\bR)}$ where here $^\dagger$ is the adjoint in $L^2(\bR, dx)$, so that $P^*=P$ and $Q^*=Q$. It is important to stress that 
we are here considering $\mathfrak{A}_{\textrm{CCR},1}$ as an abstract algebra (i.e., up to isomorphisms of unital $*$-algebras)  independently of the above
 concrete realization.
	Consider the state $\omega$
	$$\omega(A) = \int_\bR \overline{\psi_0(x)} (A \psi_0)(x) \: dx\quad A\in \gA_{\textrm{CCR},1}\:, \quad \psi_0(x)=\pi^{-\frac 14}e^{-\frac{x^2}{2}}\:.$$
	The choice of $\omega$ is evidently related with the ground state of the harmonic oscillator.
	Exploiting the uniqueness part of the GNS theorem, it is not difficult to prove that the GNS construction generated by $\omega$ leads to
	$$
		\mathsf{H}_\omega=L^2(\mathbb{R}, dx)\:, \:\:
		\pi_\omega(A) = A\spa\rest_{\cD_\omega}\:, \:\:\psi_\omega  = \psi_0 \:.
	$$
	The crucial point which differentiates the found representation of $\gA_{\textrm{CCR},1}$ from the concrete initial realization,
  is that  now $\cD_\omega \subsetneq \cS(\bR)$. Indeed,  $\cD_\omega$ 
	results to be  the dense subspace of  $L^2(\mathbb{R}, dx)$ made of all finite linear combinations 
	of {\em Hermite functions} $\{\psi_n\}_{n \in \mathbb N}$ (the eigenstates of the harmonic oscillator Hamiltonian). $\cD_\omega$ is dense just because
 $\{\psi_n\}_{n \in \mathbb N}$ is a Hilbert basis of $L^2(\mathbb{R}, dx)$.
	
	Let us pass to consider the deformations $\omega_B$, $B \in \gA_{\textrm{CCR},1}$.
	Using as equivalent generators $I$ and $A:=\frac{1}{\sqrt{2}}(X+iP)$, $A^*=\frac{1}{\sqrt{2}}(X-iP)$ -- the elements of $\gA_{\textrm{CCR},1}$ 
corresponding to the  annihilation and creation operators -- instead of $I, Q, P$  to define $\mathfrak{A}_{\textrm{CCR},1}$, one easily sees that
		$${\cal D}_\omega = {\cal D}_{\omega_B}\:, \quad \mathsf{H}_\omega=\mathsf{H}_{\omega_{B}}\:,\quad \mbox{and}\quad\pi_{\omega_B}= \pi_\omega\quad
 \mbox{for every choice of $B \in \gA_{\textrm{CCR},1}$\,.}$$
		 The first identity holds  because $\psi_\omega \in {\cal D}_{\omega_B}$ as established in  
lemma \ref{Laggiunto1} in appendix and
	 $\pi_{\omega_B}(\cdot)= \pi_\omega(\cdot)\spa\rest_{{\cal D}_{\omega_B}} $,  for every $B \in \mathfrak{A}_{\textrm{CCR},1}$ such that $\omega(B^*B)>0$ 
according to (\ref{GNSb}). 
The remaining identities are trivial consequences of the first ones.  
		 Regarding the problem of essentially selfadjointness, we have that 
\begin{itemize}
\item[(a)] $\pi_{\omega_B}(Q^k)$ (and $\pi_{\omega_B}(P^k)$) are  essentially selfadjoint 
 for $k=1,2$ because of {\em Nelson's theorem} \cite[Thm. X39]{RS2} as the $\psi_n$s are a set of {\em analytic vectors} for $\pi_{\omega_B}(Q)$ and 
$\pi_{\omega_B}(Q^2)$.
This is consequence of  the estimate, arising form
\beq
X = \frac{1}{\sqrt{2}} \left(A+A^*\right)\:, \label{XAA}
\eeq
 \beq ||\pi_{\omega}(Q^k) \psi_n|| \leq 2^{k/2} \sqrt{(n+k)!}\,,\qquad n=0,1,\ldots\,,\,k=1,2, \ldots\label{estimate0}\:,\eeq
 arising from \cite[Example 2 p.204]{RS2} which implies
	 \beq ||\pi_{\omega}(Q^k)^m \psi_n|| \leq 2^{mk/2} \sqrt{(n+mk)!}\,,\qquad m,n=0,1,\ldots\,,\,k=1,2, \ldots\label{estimate}\:;\eeq
	 and of  the fact that the span $\cD_{\omega_B}= \cD_\omega$ of the $\psi_n$s  is furthermore dense in $\sH_{\omega_B}=\sH_\omega$.
\item[(b)] $\pi_{\omega}(Q^4)=\pi_{\omega_B}(Q^4)$ is essentially selfadjoint. In fact, it is  symmetric, {\em bounded from below} and a direct computation 
based on (\ref{estimate}) proves that the $\psi_n$ are {\em semianalytic vectors} for it, hence we can apply \cite[Thm. X40]{RS2} 
which guarantees that $\pi_{\omega_B}(Q^4)$ is essentially selfadjoint.
\end{itemize}
		Let us focus on the moment problem relative to $(Q^k,\omega)$. 
Observe that, if $a$ and $a^+\subset a^\dagger$ are the standard {\em annihilation} and {\em creation operators} defined on $\cD_{\omega}$ such that \beq \mbox{$[a,a^+] = I$,\:
 $[a,a] =[a^+,a^+]= 0$, \:
and $a\psi_0 =0$,}\label{algQM}\eeq then it holds
\beq
\pi_\omega(Q) = \frac{1}{\sqrt{2}}\left( \pi_\omega(A)+ \pi_\omega(A^*)\right)=  \frac{1}{\sqrt{2}}\left( a+ a^+\right)\:. \label{Qaa}
\eeq
		Therefore, if  $s^{(k)}_n$ denotes the $n$-th moment of $Q^k$ in the state $\omega$ ,we have
		\begin{align}
			s^{(k)}_n:=\omega(Q^{kn})=\frac{1}{2^{kn/2}}\left\langle \psi_0 \left| (a+a^+)^{kn}\right.\psi_0\right\rangle  =\pi^{-\frac{1}{2}}\int_{\mathbb{R}}x^{kn}e^{-x^2}dx\,.\label{sal}
		\end{align}
		The moment problem relative to $(Q^k,\omega)$ admits at least a solution $\mu^{(Q^k)}_\omega$ due to Proposition \ref{P2} because $Q^k$ is Hermitian in the algebra.
		Let us examine uniqueness of this measure, i.e, in the jargon of moment problem theory, we go to check  if the moment problem is {\em determinate} taking advantage 
	of the results discussed in \cite{Schmudgen17}.
		\begin{itemize}
			\item[$k=1$] We may directly compute $s^{(1)}_{2n+1}=0$ and $s^{(1)}_{2n}=2^{-2n}(2n-1)!!$, which satisfies the hypothesis of {\em Carleman's condition} \cite[Cor. 4.10]{Schmudgen17}: the moment problem for $\{s^{(1)}_n\}_{n\in \bN}$ is thus {\em determinate}.
			\item[$k=2$]  We may apply {\em Cramer's condition} \cite[Cor.4.11]{Schmudgen17} to conclude that the moment problem for $\{s^{(2)}_n\}_{n\in \bN}$ is again {\em determinate}.
			\item[$k=3$]  We have
			\begin{align}
			s^{(3)}_n:=\pi^{-\frac{1}{2}}\int_{\mathbb{R}}x^{3n}e^{-x^2}dx
			=\frac{1}{3\sqrt{\pi}}\int_{\mathbb{R}}y^{n}\frac{e^{-y^\frac{2}{3}}}{y^{\frac{2}{3}}}dy
			=:\int_{\mathbb{R}}y^nf(y)dy\,.
			\end{align}
			We now apply {\em Krein's condition} for indeterminacy \cite[Thm. 4.14]{Schmudgen17}: since
			\begin{align*}
			\int_{\mathbb{R}}\frac{\log f(x)}{1+x^2}dx=
			-\int\frac{\log(3\sqrt{\pi})}{1+x^2}dx
			-\int_{\mathbb{R}}\frac{x^{\frac{2}{3}}}{1+x^2}dx
			-\frac{2}{3}\int_{\mathbb{R}}\frac{\log |x|}{1+x^2}dx>-\infty\,,
			\end{align*}
			the moment problem for $\{s^{(3)}_n\}_{n\in \bN}$ is {\em not  determinate}.
			\item[$k=4$] We have
			\begin{align*}
				s^{(4)}_n=
				\frac{1}{\sqrt{\pi}}\int_{\mathbb{R}}x^{4n}e^{-x^2}dx=
				\frac{2}{\sqrt{\pi}}\int_0^{+\infty}x^{4n}e^{-x^2}dx=
				\frac{1}{2\sqrt{\pi}}\int_0^{+\infty}y^{n}\frac{e^{-y^\frac{1}{2}}}{y^{\frac{3}{4}}}dy=:
				\int_{\mathbb{R}}y^nf(y)dy\,.
			\end{align*}
			Once again Krein's condition is satisfied:
			\begin{align*}
				\int_{\mathbb{R}}\frac{\log f(x)}{1+x^2}dx=
				-\int_0^{+\infty}\frac{\log(2\sqrt{\pi})}{1+x^2}dx
				-\int_0^{+\infty}\frac{x^{\frac{1}{2}}}{1+x^2}dx
				-\frac{3}{4}\int_0^{+\infty}\frac{\log x}{1+x^2}dx>-\infty\,.
			\end{align*}
			The moment problem associated with $\{s^{(4)}_n\}_{n\in \bN}$ is therefore  {\em not  determinate}.
		\end{itemize}
 The case $k=4$ provides  a counter-example to the converse of Theorem \ref{T1} in elementary QM.
	In fact,  there are many measures $\mu^{(Q^4)}_\omega$ associated to the pair $(\pi_\omega(Q^4),\omega)$ 
because the moment problem is indeterminate, 
but every $\pi_{\omega_B}(Q^4)= \pi_{\omega}(Q^4)$ is essentially selfadjoint.\\
 All the found results can be recast for $P^n$  by exploiting the unitary map (Fourier transform) that transforms $\pi_\omega(Q)$
to   $\pi_\omega(P)$ and $\pi_\omega(P)$
to   $-\pi_\omega(Q)$ leaving $\psi_0$ invariant.\\

\noindent {\bf (2)} Taking advantage of the results presented in \cite{Wald}, let us consider again the $*$-algebra of bosonic fields as in item (2) of example \ref{exQFT} 
whose $*$-algebra is $\gA[M,g]$
and the generators (field operators smeared with real smooth solutions of the KG equation with compactly supported Cauchy data) are denoted by $\Phi[\varphi]$.
 A particularly relevant  class of 
states $\omega : \gA[M,g] \to \bC$ is the one of {\em Gaussian} ones (also known as  {\em quasifree}, see e.g., \cite{Kay-Wald-91,Wald,book}).
If $\omega$ is Gaussian \cite{KM}, then there is a $\bR$-linear map $K\colon Sol[M,g] \to \sH^{(1)}_\omega$, where $\sH^{(1)}_\omega$ is a Hilbert space
called {\em one-particle space}, with the following properties: \begin{itemize} \item[(a)] $K(Sol[M,g])+ i K(Sol[M,g])$ is dense in $\sH^{(1)}_\omega$;
\item[(b)] for all $\varphi,\tilde{\varphi} \in Sol[M,g]$ it
 holds $\langle K\varphi,K\tilde{\varphi} \rangle=\omega(\Phi[\varphi] \Phi[\tilde{\varphi} ])$;
\item[(c)] for all $\varphi,\tilde{\varphi} \in Sol[M,g]$ it
 holds  $Im \langle K\varphi,K\tilde{\varphi} \rangle = -\frac{1}{2}\sigma(\varphi, \tilde{\varphi})$.
\end{itemize}
Moreover the pair $(K, \sH_\omega^{(1)})$ is unique up to unitary transformations and it determines the GNS structure of $\omega$ as we are going to illustrate.  From now on,  $ \cF_+(\sH_\omega^{(1)})$ is the {\em bosonic Fock space} relying upon the {\em one-particle subspace} $\sH_\omega^{(1)}$ and 
 $\sH_\omega^{(n)}\subset \cF_+(\sH_\omega^{(1)})$
is the  subspace made of all  symmetrized products of  $n=1,2,\ldots$ factors in $\sH_\omega^{(1)}$. Furthermore
$a_x$ and $a^+_x\subset (a^+_x)^\dagger$, 
with $x \in \sH^{(1)}_\omega$, are 
standard {\em bosonic annihilation} and {\em creation operators} defined on the dense invariant subspace $\cN_\omega$
given by the finite span of $\psi_\omega$ and all spaces $\sH_\omega^{(n)}$.
They satisfy,
\beq \mbox{$[a_x,a^+_y] = \langle x|y\rangle I$, \:
$[a_x,a_y]=
[a^+_x,a^+_y]=0$, \: and $a_x\psi_\omega=0$,} \label{algQFT}\eeq
where $\langle \cdot |\cdot \rangle$ is the inner product in $\sH_\omega^{(1)}$.
The GNS structure of $\omega$ is as follows.
\begin{itemize}
\item[(i)] $\sH_\omega = \cF_+(\sH^{(1)}_\omega)$.
\item[(ii)]  The representation $\pi_\omega$ has the explicit form\footnote{The creation and annihilation operators $a,a^\dagger$ appearing  in Eq.(3.2.28) in  \cite{Wald} are our
 $-ia_{K\varphi}$ and
$ia^+_{K\varphi}$.} (to compare with (\ref{XAA}))
\beq \pi_\omega(\Phi[\varphi]) =  \left( a_{K\varphi}+ a^+_{K\varphi}\right)|_{\cD_\omega}\:,\label{PHI}\eeq
 where  $\cD_\omega$ is defined below.
\item[(iii)]  $\cD_\omega\subset \cN_\omega$  is the span  of the vectors  constructed by applying the operators $a^+_{K\varphi}$, with $\varphi \in Sol[M,g]$, to the GNS cyclic
 vector $\psi_\omega$ a finite but  arbitrarily large  number of times.  
\item[(iv)]  $\psi_\omega$ coincides with the vacuum
 state 
vector
of $\cF_+(\sH^{(1)}_\omega)$. 
\end{itemize}
An estimate similar to (\ref{estimate0}) holds true \cite[Proof of Theorem X.41,  p.210]{RS2}\footnote{The field operators $\Phi_S(f)$ used there correspond to our
 $2^{-1/2}\pi_\omega(\Phi[\varphi])$.}
\beq ||a_x \psi||\;, ||a^+_x \psi|| \leq \sqrt{(n+1)} ||x||||\psi|| \,,\qquad   x \in \sH_\omega^{(1)}\:, \psi \in \sH^{(n)}_\omega\:,   n=0,1,\ldots
\label{estimateQFT0}\:,\eeq
and also, for $\varphi \in Sol[M,g]$ and  $n=0,1,\ldots\,,\,k=1,2, \ldots$,
\beq || \pi_\omega((2^{-1/2} \Phi[\varphi])^k ) \psi|| \leq ||K\varphi||^k 2^{k/2}\sqrt{(n+k)!}||\psi|| \,,\qquad \psi \in \sH_\omega^{(n)}\cap \cD_\omega,
\label{estimateQFT}\eeq
as a consequence of (\ref{PHI}).\\
 Let us pass to discuss the GNS quadruple of deformations $\omega_b$ with the hypothesis that $\omega$ is pure.
It turns out that $\omega$ is {\em pure if and only if} $K(Sol[M,g])$ alone is dense in $\sH^{(1)}$ \cite{Kay-Wald-91}.
In that case, if  $\varphi \in Sol[M,g]$, there must be a sequence  $\varphi_n \in Sol[M,g]$
such that  $K\varphi_n \to iK\varphi$ for $n\to +\infty$.
Since $a_{x}$ and $a^+_{x}$ are respectively antilinear and linear  in their arguments $x \in \sH_\omega^{(1)}$, we have from (\ref{PHI})
 and (\ref{estimateQFT0}),
\beq a_{K\varphi}\psi = \lim_{n\to +\infty}\frac{1}{2} \left( \pi_\omega( \Phi[\varphi]) + i  \pi_\omega( \Phi[\varphi_n]) \right)\psi \quad \mbox{and}\quad
 a^+_{K\varphi}\psi = \lim_{n\to +\infty}\frac{1}{2} \left( \Phi[\varphi] ) - i  \pi_\omega( \Phi[\varphi_n] ) \right)\psi\nonumber \eeq
for every given   $\psi \in \cD_\omega$.
 Using  these  identities and taking estimates (\ref{estimateQFT0}) into account for passing from iterated limits in $n$
 to single limits in $n$, one sees that  $\pi_\omega(\gA[M,g]) \pi_\omega(b) \psi_\omega$ 
is dense in $\sH_\omega$ for every given 
$b\in \gA[M,g]$ such that $\omega(b^*b) \neq 0$. We conclude that, if $\omega$ is pure, the GNS 
structure of the (non-singular) deformation 
$\omega_b$ is 
$$\cD_{\omega_b}=\pi_\omega(\gA(M,g)) \pi_\omega(b) \psi_\omega\:,\quad  \sH_{\omega_b}=\sH_\omega =\cF_+(\sH_\omega^{(1)}) \:, \quad \pi_{\omega_b}= \pi_\omega|_{\cD_{\omega_b}}\:, \quad \psi_{\omega_b}=\pi_\omega(b)\psi_\omega.$$
In particular $\cD_{\omega_b}\subset \cD_\omega$.
There are however special physically important  cases where $\cD_{\omega_b} = \cD_\omega$ so that all the discussion about the CCR algebra of $X$ and $P$  in the previous example, including the 
counterexamples, 
can be  completely recast in this QFT context.
The form above of the GNS quadruple of $\omega_b$ implies that $\cD_{\omega_b} = \cD_\omega$ is equivalent to $\psi_\omega \in \cD_{\omega_b}$.
A sufficient condition for  $\psi_\omega \in \cD_{\omega_b}$  is that, 
\beq \mbox{for every $\varphi \in Sol[M,g]$ there is $\varphi' \in Sol[M,g]$ such that  $K\varphi' = i K\varphi$} \:.\label{iK}\eeq
This requirement  implies in particular that the state $\omega$ is pure and also promotes $a_{K\varphi}$ and $a^+_{K\varphi}$
 to elements and  generators of the algebra $\gA[M,g]$ as in the case of the CCR
of $X$ and $P$. Everything is established in lemma \ref{Laggiunto2} in appendix \ref{AppendixB}.
 When (\ref{iK}) is valid,  the GNS structure of the (non-singular) deformation 
$\omega_b$ is therefore  again
\beq \cD_{\omega_b} =\cD_\omega\:,\quad  \sH_{\omega_b}=\sH_\omega =\cF_+(\sH_\omega^{(1)}) \:, \quad \pi_{\omega_b}= \pi_\omega\:, \quad \psi_{\omega_b}=
\pi_\omega(b)\psi_\omega.\label{GNSQFT}\eeq
The validity of  (\ref{iK})  is in particular guaranteed if $(M,g)$ is the standard {\em four dimensional Minkowski spacetime} and $\omega$ is the (pure) 
{\em Poincar\'e invariant 
Gaussian vacuum state} for the Klein-Gordon quantum field with {\em strictly positive} mass $m>0$. Here,  the Klein-Gordon operator reads $P= -\frac{\partial^2}{\partial t^2} + \Delta -m^2$
 where $\Delta$
is the standard spatial Laplacian in $\bR^3$ and we are using a given Minkowskian coordinate system $(t, {\bf x)} \in \bR \times \bR^3$ whose $t=0$ surface is the preferred
 Cauchy surface $\Sigma$ we will henceforth use.
In this case,  the general Fock-space  structure of the GNS representation of $\omega$
declared in (a)-(c), (i)-(iv)
 is valid also redefining the space $Sol[M,g]$ as the space of real smooth solutions of KG equation 
with Cauchy data (in particular on $\Sigma$) which belongs to   $ \cS(\bR^3)_\bR$ (the real space of real-valued Schwartz functions on 
$\bR^3$). It turns out that $\sH_\omega^{(1)} = L^2(\bR^3, d{\bf k})$ and
$$\left(K\varphi\right)({\bf k}) = \frac{1}{\sqrt{2}{(\bf k}^2 +m^2)^{1/4}}\int_{\bR^3} \frac{e^{-i {\bf k}\cdot {\bf x}}}{(2\pi)^{3/2}}
 \left( 
\sqrt{-\Delta+ m^2} \varphi(0,{\bf x}) + i \partial_t \varphi(0,{\bf x})   \right)  d{\bf x}\:, \quad \varphi \in Sol[M,g] \:.$$
The crucial observation is that, if $\varphi(0,\cdot), \partial_t\varphi(0,\cdot) \in \cS(\bR^3)_\bR$ then 
$$f := - (-\Delta +m^2)^{-1/2} \partial_t\varphi(0,\cdot) \:, \quad p := (-\Delta +m^2)^{1/2} \varphi(0,\cdot)$$
still satisfy   $f,p \in \cS(\bR^3)_\bR$ so that there is a unique $\varphi' \in Sol[M,g]$ with Cauchy conditions $\varphi'(0, \cdot) = f$ and $\partial_t\varphi'(0,\cdot) = p$ and due to the above expression for $K$,   $\varphi'$ satisfies (\ref{iK}).\\

In summary, for a state $\omega : \gA[M,g] \to \bC$,
we finally have the following results analogous to the ones in the quantum mechanical case of (1) example \ref{Ex: uniqueness of moment problem and self-adjointness of the operator are not related} and established with an essentially identical procedure
\footnote{Condition  \eqref{iK} has been verified for the algebra $\mathfrak{A}[M,g]$ when  $(M,g)$ is 4D Minkowski  spacetime and  the  generators $\Phi[\varphi]$ are smeared  with real test Schwartz functions serving as Cauchy data. This is a very special case which is suitable in Minkowski spacetime especially because therein the Cauchy surfaces can be chosen as submanifolds 
isometrically isomorphic to $\bR^3$. However,  we expect that the conclusions (a1), (b1) hold  true also for the algebra $\mathfrak{A}(M,g)$ (i.e., field operators smeared with Cauchy data in $C_c^\infty(\bR^3)$)  as the closure of the $\pi_\omega$-representation of the generators $\Phi(f)\in\mathfrak{A}(M,g)$ and $\Phi[Ef]\in\mathfrak{A}[M,g]$ coincides, though a rigorous proof of this fact is not presented here.
Some standard spacetime-deformation argument  would probably allow to prove similar conclusions for a number of pure quasifree states in curved spacetimes too, though we refrain to specify the precise set of states.
We are grateful to an anonymous referee for drawing our attention on these issues.}.
\begin{itemize}
\item[(a)]  $\pi_\omega(\Phi[\varphi]^k)$ are essentially selfadjoint  for $k=1,2$ since  the vectors in  $\sH_\omega^{(n)}\cap 
\cD_\omega$ for all $n$, whose span is dense in $\sH_\omega$,
are  analytic vectors for $k=1,2$ for those symmetric operators.
This is consequence of the inequality  arising from (\ref{estimateQFT}), for  $n=0,1,\ldots\,,\,m, k=1,2, \ldots $
\beq ||\pi_{\omega}(((2^{-1/2}\Phi[\varphi])^k)^m \psi|| \leq (2||K\varphi||^2)^{mk/2} \sqrt{(n+mk)!}||\psi|| \,,\quad \psi \in \sH_\omega^{(n)}\cap \cD_\omega, \:
\label{estimateQFT3}\eeq
\item[(a1)] When assuming (\ref{iK}), the result in (a) extends to  
$\pi_{\omega_b}(\Phi[\varphi]^k)$, for every $b\in \gA[M, g]$,
  just in view of the the GNS structure (\ref{GNSQFT}), by  noticing that $\pi_{\omega_b}(\Phi[\varphi]^k)= \pi_\omega(\Phi[\varphi]^k)$, $\sH_{\omega_b}= \sH_\omega$ and
 $\sH_\omega^{(n)}\cap 
\cD_\omega = \sH_{\omega}^{(n)} \cap \cD_{\omega_b}$.
\item[(b)] $\pi_\omega(\Phi[\varphi]^4)$ is essentially selfadjoint because $\pi_\omega(\Phi[\varphi]^4)$ 
is bounded below -- still using \eqref{estimateQFT3} --
 while the vectors in $\cup_n(\sH_\omega^{(n)}\cap \cD_\omega) = \cup_n(\sH_\omega^{(n)} \cap \cD_{\omega_b})$ are semi-analytic vectors for $\pi_\omega(\Phi[\varphi]^4)$
 and their span is dense in $\mathsf{H}_\omega$.
\item[(b1)] When assuming (\ref{iK}), the result in (b) extends to  
 $\pi_{\omega_b}(\Phi[\varphi]^4), $for every $b\in \gA[M, g]$,
  just in view of the the GNS structure (\ref{GNSQFT}), by  noticing that $\pi_{\omega_b}(\Phi[\varphi]^4)= \pi_\omega(\Phi[\varphi]^4)$, $\sH_{\omega_b}= \sH_\omega$ and
 $\sH_\omega^{(n)}\cap 
\cD_\omega = \sH_{\omega}^{(n)} \cap \cD_{\omega_b}$.
\end{itemize}
Concerning the moment problem of $\pi_\omega(\Phi[\varphi]^k)$ with $K\varphi\neq 0$,  the sequence of moments
\begin{align}\label{momQFT}
			s^{(k)}_n:=\omega((2^{-1/2}\Phi[\varphi])^{kn})=\frac{1}{2^{kn/2}}\left\langle \psi_\omega \left| (a_{K\varphi}+a_{K\varphi}^+)^{kn}\right.\psi_\omega\right\rangle\:, 
		\end{align}
is identical to the analogous sequence (\ref{sal}) of $\pi_\omega(Q^k)$ discussed in the previous example just because 
(\ref{algQM}) and (\ref{algQFT}) are formally identical (for $||K\varphi||=1$ otherwise it is sufficient to normalize $\varphi$ accordingly) and the moments are just computed using only them as is evident if comparing (\ref{sal}) and (\ref{momQFT}).
Therefore, for $k=1,2$ the problem is determinate, whereas it is not determinate for $k=3,4$.\\
Assuming (\ref{iK}), thus in Minkowski spacetime in particular,   $k=4$ provides  counter-examples to the converse of Theorem \ref{T1} in elementary QFT:
 there are many measures $\mu^{(\Phi[\varphi]^4)}_\omega$ associated to the pair $(\Phi[\varphi]^4,\omega)$ 
because the moment problem is indeterminate,  but every $\pi_{\omega_b}(\Phi[\varphi]^4)= \pi_{\omega}(\Phi[\varphi]^4)$ is essentially selfadjoint.\hfill $\blacksquare$}\\
\end{example}
\begin{remark}
	{\em Since $\pi_\omega(Q^4)\geq 0$, we  can try to restrict our analysis  of  existence and uniqueness
problem  for measures  solving the moment problem for the sequence $\{s^{(4)}_n\}_{n \in \bN}$
when they are supported in 
$[0,+\infty)$ rather than in  the whole $\bR$.
This alternate   formulation  is called  {\em Stieltjes moment  problem}. 
				However, if this problem were determinate with unique measure $\mu$, the standard
 Hamburger problem would be determine as well (but we know that it is not)
unless  $\mu(\{0\})\neq 0$ on account of \cite[Corollary 8.9]{Schmudgen12}. Since $\overline{\pi_\omega(Q^4)}$ is 
 selfadjoint, that unique  $\mu$ would also  with the measure (\ref{muPVM})
obtained from the PVM of $\overline{\pi_\omega(Q^4)}$  with $\bR$ replaced for $[0,+\infty)$ since also this spectral 
measure is a solution of the same moment problem over $[0,+\infty)$.
 On the other hand, $\overline{\pi_\omega(Q^4)}$   has empty point spectrum (it is the multiplicative operator $x^4$ in $L^2(\bR,dx)$), against the 
assumption $\mu(\{0\})\neq 0$.
Therefore  also Stieltjes problem is  
{\em not determinate}. A similar comment can be ascribed to $\Phi[\varphi]^4$.}\hfill $\blacksquare$\\
\end{remark}

\begin{remark}
	{\em Example \ref{Ex: uniqueness of moment problem and self-adjointness of the operator are not related} shows that powers of a quantum field $\Phi(h)^k$ may lead to algebraic observables with a possibly non-determinated moment problem.
	It would be interesting to see whether a similar result holds for the Wick powers $\colon\spa\Phi^k\spa\colon(h)$ of a quantum field and in particular for the stress-energy tensor.
	Results on the self-adjointness of these observables are already present in the literature \cite{Rabszytn-89,Sanders-12}.
	We postpone the discussion of this example to a future investigation.} \hfill $\blacksquare$  \\
\end{remark}

\noindent Corollary \ref{Cor: uniqueness of consistent measures} below  can be in a sense interpreted  as a weak converse of Theorem \ref{T1}. However, to see it,  a suitable mathematical technology must be introduced.

\section{The notion of POVM and its relation with symmetric operators}\label{Sec: The notion of POVM and its relation with symmetric operators}
A {\em Positive Operator Valued Measure} (POVM for short) is an extension of the notion of {\em Projection Valued Measure} (PVM).
Since PVMs are one-to-one with selfadjoint operators and have  the physical meaning of a quantum observable (see, e.g., \cite{Moretti2017} for a thorough discussion on the subject), POVMs provide a generalization of the notion of quantum observable.
Similarly to the fact that PVMs are related with selfadjoint operators, it results that POVMs are connected to merely {\em symmetric} operators, even if this interplay is more complicated.
Since GNS operators $\pi_\omega(a)$ representing Hermitian elements are in general only symmetric, the notion of POVM seems to be relevant in
 our discussion on Issue A.

We briefly collect below some material on POVMs  and generalized extension of symmetric operator -- see \cite{Akniezer-Glazman-93,Busch-14, Dubin-Kiukas-Pellonpaa-Ylinen-14} for a complete discussion. 

\begin{remark}\label{remequ}
	{\em The complete equivalence between  the notion of  POVM used in 
 \cite{Busch-14,Dubin-Kiukas-Pellonpaa-Ylinen-14} and the older notion  of {\em spectral function} adopted in \cite{Akniezer-Glazman-93}   is  discussed and established in Section 4.9 of \cite{ Busch-14}, especially Theorem 4.3 therein.
 In  \cite{ Busch-14}, spectral functions are called {\em semispectral functions} while  normalized POVMs are  named {\em semispectral measures}.}\hfill $\blacksquare$
\end{remark}

\subsection{POVM as a generalized observable in a Hilbert space}
 $(\Omega,\Sigma)$ will henceforth denote a measurable space, where $\Sigma$ is a $\sigma$-algebra on the set $\Omega$.
$\mathfrak{B}(\mathsf{H})$ will denote the space of bounded linear operators on the Hilbert space $\mathsf{H}$ and 
 $\cL(\mathsf{H})\subset\mathfrak{B}(\mathsf{H})$ is the space of orthogonal projections over $\mathsf{H}$. We start from the following general definition which admits some  other equivalent formulations \cite{Dubin-Kiukas-Pellonpaa-Ylinen-14} (see Remark 4.4 of \cite{Busch-14} in particular, where the second requirement below is alternatively and equivalently stated).\
\begin{definition}\label{Definition: POVM}
	{An operator-valued map $Q\colon\Sigma\to\mathfrak{B}(\mathsf{H})$ is called {\bf positive-operator valued measure} (POVM)   if it satisfies the following two conditions:
	\begin{enumerate}
		\item
		for all $E\in\Sigma$, $Q(E)\geq 0$;
		\item 
		for all $\psi,\varphi\in\mathsf{H}$
 the map $Q_{\psi,\varphi}\colon\Sigma\ni E\mapsto\langle \psi|Q(E)\varphi\rangle\in\mathbb{C}$ 
defines a {\em complex $\sigma$-additive  measure}  according to \cite{Rudin87}
\footnote{
	With reference to \cite[Ch. 6]{Rudin87} this implies that for all $E\in\mathscr{B}(\Sigma)$ and for all countable partitions $E=\cup_{n\in\mathbb{N}} E_n$, with $E_n\in\mathscr{B}(\Sigma)$ and $E_n\cap E_m=\emptyset$ if $n\neq m$, it holds $Q_{\psi,\varphi}(E)=\sum_{n\in\mathbb{N}}Q_{\psi,\varphi}(E_n)$, where the series is absolutely convergent. 
	In particular, the {\em total variation}  of this measure $|Q_{\psi, \varphi}|$ is a {\em finite} positive $\sigma$-additive measure.
	Notice that
	$Q_{\psi,\psi}$ is a finite positive measure.}.
	\end{enumerate}
A POVM $Q$ is said  to be {\bf normalized} if $Q(\Omega)=I$.}\hfill $\blacksquare$\\
\end{definition}
\begin{remark}
{\em 	A normalized POVM 
is a standard  
  PVM (e.g. see \cite{Moretti2017,Schmudgen12}) if and only if  $Q(E)Q(F)=Q(E\cap F)$ for  $E,F\in \Sigma$,   so that 
$Q(E)\in\cL(\mathsf{H})$.}\hfill $\blacksquare$\\
\end{remark}
\noindent 
Normalized POVMs are physically intepreted and called {\bf generalized observables}, in the Hilbert-space formulation of quantum theory and we henceforth adopt this interpretation and apply it to our context.
Since $Q(E)\geq 0$ for every POVM, the map   $Q_{\psi,\psi}\colon \Sigma\ni E\mapsto\langle \psi|Q(E)\psi\rangle$ always defines a
 {\em finite} 
{\em positive} $\sigma$-additive measure for every fixed  $\psi\in\mathsf{H}$ which is also a {\em probability measure} on $(\Omega, \Sigma)$ if $Q$ is normalized and $||\psi||=1$.
Similarly to  what happens for a PVM, the physical interpretation of $\langle \psi|Q(E)\psi\rangle$ is the probability that, measuring the {\em generalized observable} associated to the normalized  POVM when the state is represented by the normalized vector 
$\psi$, the outcome belongs to the 
Borel set $E\subset \bR$. 

What is lost  within  this  more general framework in comparison with the physical interpretation of PVMs  is (a) the logical interpretation of $Q(E)$ as an  elementary {\em YES-NO observable} also known as {\em test},  (b) the possibility to describe the {\em post-measurement} state with the standard L\"uders-von Neumann reduction postulate 
exploiting only the POVM (more information must be supplied), (c) the fact that  observables $Q(E)$ 
and $Q(F)$ are necessarily compatible.

There exists an extended  literature on these topics and we refer the reader to \cite{Busch-14} for a modern also 
physically minded treatise on the subject.
Another difference concerns the one-to-one  correspondence  between PVM over $\bR$ and selfadjoint operators which, in the standard spectral theory, permits to identify PVMs (quantum observables) with selfadjoint operators.

Switching to  POVMs, it turns out that there is a more complicated correspondence between normalized POVMs and {\em symmetric operators} which we will describe shortly. 
The typical generalized observable which can be described in terms of a POVM is  the (arrival) {\em time observable} of a particle \cite{Busch-14}. That observable cannot be described in terms of selfadjoint operators (PVMs) if one insists on the validity of CCR with respect to the energy observable and these no-go results are popularly known as {\em Pauli's theorem} (see, e.g. \cite{Moretti2019}).

A celebrated result due to Naimark establishes  that POVMs are connected to PVMs  through  the famous  {\em Naimark's dilation theorem}, which we state for the case of a normalized POVM  \cite[Thm. Vol II, p.124]{Akniezer-Glazman-93} (see \cite{Dubin-Kiukas-Pellonpaa-Ylinen-14} for the general case).

\begin{theorem}{\bf [Naimark's dilation theorem]}\label{Theorem: Naimark's dilation theorem}
	Let $Q\colon\Sigma\to\mathfrak{B}(\mathsf{H})$ be a normalized POVM.
	Then there exists a Hilbert space $\mathsf{K}$ which includes $\mathsf{H}$ as a closed subspace, i.e. $\mathsf{K}= \mathsf{H}\oplus \mathsf{H}^\perp$, and a PVM $P\colon\Sigma\to\cL(\mathsf{K})$ such that
	\begin{align}\label{Equation: Naimark's dilation relation}
		Q(E)=P_{\mathsf{H}}P(E) \spa\rest_{\mathsf{H}}\qquad\forall E\in\Sigma\:,
	\end{align}
	where $P_{\mathsf{H}}\in \cL(\mathsf{K})$ is the orthogonal projector onto $\mathsf{H}$.
	The triple $(\mathsf{K},P_{\mathsf{H}},P)$ is called {\bf Naimark's dilation triple}.
\end{theorem}

\begin{remark}
{\em 	A slightly more general  way to state  the theorem above, is stating that, for a normalized POVM $Q\colon\Sigma\to\mathfrak{B}(\mathsf{H})$,
 there exist a Hilbert space $\mathsf{K}_0$ and an isometry  $V: \mathsf{K}_0 \to \mathsf{K}$ such that
	\begin{align}\label{Equation: Naimark's dilation relation2}
		Q(E)=V^\dagger P(E)V\qquad\forall E\in\Sigma\:.
	\end{align}
In this way, $V^\dagger V=I_{\mathsf{K}_0}$ and  $VV^\dagger\in \cL(\mathsf{K})$ is the orthogonal projector onto  the  closed subspace  $V(\mathsf{K}_0) \subset \sK$. 
Within this formulation,  {\em Naimark's dilation triple} is defined as $(\mathsf{K},V,P)$. In  (\ref{Equation: Naimark's dilation relation}), $\mathsf{K}_0 = \mathsf{H}$ and $V$ is the inclusion map $\mathsf{H}\hookrightarrow \mathsf{K}$, so that $P_{\mathsf{H}} = VV^\dagger$. }\hfill $\blacksquare$
\end{remark}
\subsection{Generalized selfadjoint extensions of symmetric operators} POVMs arise naturally when dealing with generalized extensions of {\em  symmetric operators}.
As is well known, a selfadjoint operator $A$ in a Hilbert space $\mathsf{H}$ does not admit proper symmetric extensions in $\mathsf{H}$. 
This is just a case of a more general class of symmetric operators.

\begin{definition}\label{Definition: maximally symmetric operator}
	{\em A  symmetric operator $A$  on a Hilbert space $\mathsf{H}$ is said to be {\bf maximally symmetric} if there is no symmetric operator $B$ on  $\mathsf{H}$ such that $B\supsetneq A$.}\hfill $\blacksquare$\\
\end{definition}

\begin{remark} $\null$\\
{\em  {\bf (1)} A maximally symmetric operator is necessarily closed, since the closure of a symmetric operator is symmetric as well.\\
{\bf (2)} It turns out that  \cite[Thm. 3, p.97]{Akniezer-Glazman-93}  {\em a closed symmetric operator is maximally symmetric (and not 
selfadjoint) iff one of its deficiency indices is $0$ (and the other does not vanish)}.\\ 
{\bf (3)} An elementary useful  result (immediately arising from, e.g. \cite[Thm. 5.43]{Moretti2017}) is that, {\em if a maximally symmetric operator $A$ 
on $\mathsf{H}$ satisfies $CA\subset AC$ for a conjugation $C: \sH \to \sH$, then $A$ is selfadjoint}. }\hfill $\blacksquare$\\\end{remark}

\noindent Symmetric operators can also admit
 extensions in a more general fashion and these extensions play a crucial r\^ole in the connection between symmetric operators and POVMs.
\begin{definition}\label{Definition: generalized extension of symmetric operator}
	 {\em Let $A$ be a  symmetric operator on a Hilbert space $\mathsf{H}$.
A {\bf generalized symmetric} (resp. {\bf selfadjoint}) {\bf extension} of $A$ is a symmetric (resp. selfadjoint) operator $B$ on a Hilbert space $\mathsf{K}$ such that  
\begin{itemize}
\item[(i)] $\mathsf{K}$ contains $\mathsf{H}$ as a closed subspace (possibly $\mathsf{K}= \mathsf{H}$),
 \item[(ii)] $A\subset B$ in $\mathsf{K}$,
\item[(iii)] every closed subspace $\mathsf{K}_0 \subset  \mathsf{K}$ 
such that $\{0\} \neq  \mathsf{K}_0 \subset  \mathsf{H}^\perp$
 does not {\em reduce} $B$ (see Appendix \ref{reduction}).
\end{itemize}}\hfill $\blacksquare$\\
\end{definition}
\noindent Every non-selfadjoint symmetric operator (possibly maximally symmetric) always admits generalized selfadjoint extensions as established 
in Theorem \ref{theoremkindII}. 
Selfadjoint operators are instead maximal also in respect of  this more general sort of extension.
\begin{proposition}\label{prop1}
	A selfadjoint operator  does not admit proper generalized symmetric extensions.
\end{proposition}

\begin{proof} See Appendix \ref{AppendixB}.
\end{proof}

\subsection{Decomposition of symmetric operators in terms of POVMs}  
Naimark extended part of the spectral theory usually formulated in terms of PVMs for normal closed operators
 (selfadjoint in particular)  to the more general case of a symmetric operator \cite{Naimark1940,Naimark1940b}
where  POVMs replace PVMs. A difference with the standard theory  is  that, unless the symmetric operator is {\em maximally symmetric}, the  POVM which decomposes it is not unique.

\begin{theorem}\label{Theorem: uniqueness of POVM associated with maximally symmetric operator}
	For a symmetric operator $A$  in the Hilbert space $\mathsf{H}$ the following facts hold.
	\begin{itemize}
	\item[{\bf (a)}] There exists a normalized POVM  $Q^{(A)}\colon\mathscr{B}(\mathbb{R})\to\mathfrak{B}(\mathsf{H})$ satisfying 
	\begin{align}
		\langle\psi|A\varphi\rangle=\int_{\mathbb{R}}\lambda dQ^{(A)}_{\psi,\varphi}(\lambda)\:, \quad
		\|A\varphi\|^2=\int_{\mathbb{R}}\lambda^2dQ^{(A)}_{\varphi,\varphi}(\lambda)\,,\quad\forall\psi\in\mathsf{H},\varphi\in D(A)\,,\label{2Equation: relations between POVM and symmetric operator}
	\end{align}
	\item[{\bf (b)}] 
	
	Every normalized POVM  $Q^{(A)}\colon\mathscr{B}(\mathbb{R})\to\mathfrak{B}(\mathsf{H})$ satisfying 
	(\ref{2Equation: relations between POVM and symmetric operator})  is of the form 
	$$
			Q^{(A)}(E):=P_{\mathsf{H}}P(E)\spa\rest_{\mathsf{H}}\qquad\forall E\in\mathscr{B}(\mathbb{R})\,.
		$$
	for some Naimark's dilation triple
	$(\mathsf{K},P_{\mathsf{H}} ,P)$ of $Q^{(A)}$  arising from  a generalized selfadjoint extension $B= \int_{\mathbb{R}} \lambda dP(\lambda)$ of $A$ in $\sK$,
	\begin{gather}
		A= B\spa\rest_{D(A)}
		\quad \mbox{and} \quad D(A)\subset\{\psi\in\mathsf{H}|\; \lambda \in L^2(\mathbb{R},Q^{(A)}_{\psi,\psi})\}\:.
	\end{gather}

\item[{\bf (c)}] A  normalized POVM  $Q^{(A)}$ satisfying 
	(\ref{2Equation: relations between POVM and symmetric operator}) is a PVM if and only if the selfadjoint operator 
 $B$ constructed out of Naimark's dilation triple of $Q^{(A)}$ as in (b) can be chosen as a standard selfadjoint extension of 
$A$.

	\item[{\bf (d)}]  If $A$ is closed,  a  normalized  POVM $Q^{(A)}$ as in 
	(\ref{2Equation: relations between POVM and symmetric operator}) exists
 that, referring to (b),  also satisfies 
\begin{gather}
	A= B\spa\rest_{D(B) \cap \mathsf{H}}
	\quad \mbox{and} \quad D(A)=\{\psi\in\mathsf{H}|\; \lambda \in L^2(\mathbb{R},Q^{(A)}_{\psi,\psi})\}\:.\label{234}
\end{gather}
	\item[{\bf (e)}]  If $A$ is closed, then $A$ is maximally symmetric if and only if  there is a unique  normalized POVM $Q^{(A)}$ as in 
	(\ref{2Equation: relations between POVM and symmetric operator}).
	In this case, (\ref{234}) is valid for 
all choices of
	$(\mathsf{K},P_{\mathsf{H}} ,P)$ generating $Q^{(A)}$
	as in  (b).
	\item[{\bf (f)}]  
 If $A$ 
	is selfadjoint, there is a unique 
  normalized  POVM $Q^{(A)}$ satisfying 
	(\ref{2Equation: relations between POVM and symmetric operator}),  and it is  a PVM. 
		 In this case $\mathsf{K}= \mathsf{H}$, $Q^{(A)}=P$, and $A=B$  for all choices of 
	$(\mathsf{K},P_{\mathsf{H}} ,P)$ generating $Q^{(A)}$
	as in  (b).
	\end{itemize}
\end{theorem}

\begin{proof}
 See Appendix \ref{AppendixB}. 
\end{proof}

\begin{corollary}\label{corollaryessselfadj} 
Let  $A$ be a  symmetric operator in $\sH$. Then
\begin{itemize}

\item[{\bf (a)}] $A$ and $\overline{A}$  admits the same class of POVMs  satisfying (a) of Theorem \ref{Theorem: uniqueness of POVM associated with maximally symmetric operator}
for $A$ and $\overline{A}$ respectively.

\item[{\bf (b)}]    $A$ admits a unique normalized POVM as in (a) of Theorem \ref{Theorem: uniqueness of POVM associated with maximally symmetric operator}
 if and only $\overline{A}$
is maximally symmetric.
In this case $$D(\overline{A})=\{\psi\in\mathsf{H}|\; \lambda \in L^2(\mathbb{R},Q^{(A)}_{\psi,\psi})\}\:.$$
\item[{\bf (c)}] The unique normalized  POVM as in (b) is a PVM if  $A$ is also essentially selfadjoint.
\end{itemize}
\end{corollary}

\begin{proof} Every  generalized selfadjoint extension of $\overline{A}$ is a  generalized extension of $A$,  since $A \subset \overline{A}$.
 Every  generalized selfadjoint extension of $A$ is closed (because selfadjoint) so that it is also  a generalized selfadjoint extension of $\overline{A}$. In view of (b)
 of Theorem \ref{Theorem: uniqueness of POVM associated with maximally symmetric operator}, $A$ and $\overline{A}$ have the
 same class of associated POVMs satisfying (a) of that theorem. Therefore $A$ admits a unique POVM if and only if $\overline{A}$
is maximally symmetric as a consequence of (e) and the identity regarding $D(\overline{A})$ is valid in view of (d)   of Theorem \ref{Theorem: uniqueness of POVM associated with maximally symmetric operator}.
 Finally,  this POVM is a PVM if  $A$ is also essentially selfadjoint due to (f) Theorem \ref{Theorem: uniqueness of POVM associated with maximally symmetric operator}.
\end{proof}

\begin{definition}\label{DEFassociate}  {\em If $A$ is a symmetric operator in the Hilbert space $\mathsf{H}$, a normalized POVM $Q^{(A)}$ over the Borel $\sigma$-algebra over $\mathbb{R}$ which satisfies 
 (a) of Theorem \ref{Theorem: uniqueness of POVM associated with maximally symmetric operator}, i.e.
\begin{align*}
		\langle\psi|A\varphi\rangle=\int_{\mathbb{R}}\lambda dQ^{(A)}_{\psi,\varphi}(\lambda)\:, \qquad
		\|A\varphi\|^2=\int_{\mathbb{R}}\lambda^2dQ^{(A)}_{\varphi,\varphi}(\lambda)\,,\qquad\forall\psi\in\mathsf{H},\varphi\in D(A)\,,
	\end{align*}
is said to be {\bf associated} to $A$ or, equivalently, to {\bf decompose} $A$.} \hfill $\blacksquare$
\end{definition}

\subsection{Hermitian operators as integrals of POVMs}  While a  symmetric operator admits at least one  normalized POVM 
which decomposes it according to Definition \ref{DEFassociate},  not all normalized POVMs decompose symmetric operators. 
The main obstruction comes from the second equation in 
\eqref{2Equation: relations between POVM and symmetric operator} as well as from the difficulty to identify a convenient notion of operator integral with respect to a POVM.
This aspect of POVMs has been investigated in \cite{Dubin-Kiukas-Pellonpaa-Ylinen-14} (see also \cite{Busch-14} for further physical comments)  in wide generality. 
We only state and prove an elementary result which, though it is not explicitly stated in \cite{Dubin-Kiukas-Pellonpaa-Ylinen-14}, it is however part of the results discussed therein.  In particular, every  POVM over $\mathbb{R}$ can be weakly integrated determining a {\em unique} Hermitian operator over a natural domain. It is worth stressing that the result strictly depends on the choice of this 
domain and different alternatives are possible in principle  \cite{Busch-14,Dubin-Kiukas-Pellonpaa-Ylinen-14} (see also (1) Remark \ref{remfA} below).\\

\begin{theorem}\label{TheoremPOVMtoA}
	If $Q\colon\mathscr{B}(\mathbb{R})\to\mathfrak{B}(\mathsf{H})$ is a normalized POVM in the Hilbert space $\mathsf{H}$, define  the 
subset $D(A^{(Q)})\subset \mathsf{H}$,
	\begin{equation}\label{defD2}
		D(A^{(Q)}) := \left\{ \psi \in \mathsf{H} \left|  \int_{\mathbb{R}} \lambda^2 dQ_{\psi,\psi}(\lambda)<+\infty \right.\right\}\:.
	\end{equation}
	The following facts are valid.
	\begin{itemize}
	\item[{\bf (a)}]  $D(A^{(Q)})$ is a subspace of  $\mathsf{H}$ (which is not necessarily dense or  non-trivial).
	\item[{\bf (b)}]  There exists a unique operator  $A^{(Q)}: D(A^{(Q)}) \to \mathsf{H}$ such that
	\begin{equation}\label{eqfixA} \langle \varphi|A^{(Q)}\psi\rangle =  \int_{\mathbb{R}} \lambda dQ_{\varphi,\psi}(\lambda)\:, \quad \forall \varphi \in \mathsf{H}\:, \forall \psi \in D(A^{(Q)})\:.
	\end{equation}
	\item[{\bf (c)}] $A^{(Q)}$ is Hermitian,  so that  $A^{(Q)}$ is  symmetric if and only if $D(A^{(Q)})$ is dense.
	\item[{\bf (d)}] If  $(\mathsf{K}, P_{\mathsf{H}}, P)$ is a Naimark's dilation triple of $Q$, then 
	\begin{equation}
		A^{(Q)}\psi = P_{\mathsf{H}} \int_{\mathbb{R}} \lambda dP(\lambda) \psi\:, \quad \forall \psi \in D(A^{(Q)})\:.
	\end{equation}
	\item[{\bf (e)}] If  there exists  a Naimark's dilation triple $(\mathsf{K}, P_{\mathsf{H}}, P)$  of $Q$ such that 
	\beq \int_{\mathbb{R}} \lambda dP(\lambda)( D(A^{(Q)})) \subset \mathsf{H}\:,\label{aggni}\eeq then $A^{(Q)}$ is closed and
	\begin{equation} \label{v234}
		||A^{(Q)}\psi||^2 =  \int_{\mathbb{R}} \lambda^2 dQ_{\psi,\psi}(\lambda)\:, \quad \forall  \psi \in D(A^{(Q)})\:.
	\end{equation}
	\end{itemize}
\end{theorem}

\begin{proof}
 See Appendix \ref{AppendixB}. 
\end{proof}

\begin{remark} \label{remfA} $\null$ \\
{\em
 {\bf (1)} Even if $A^{(Q)}$  is symmetric, we cannot say that $Q$ decomposes $A^{(Q)}$ according to Def. \ref{DEFassociate}, because 
generally  (\ref{v234}) fails. However, it is still possible that  (\ref{v234}) holds when restricting $A^{{(Q)}}$ to a subspace
 $D\subset D(A^{(Q)})$.  This is the case if (\ref{eqfixA}) is true  for $D(A^{(Q)})$ replaced for $D$.
The operator $A^{(Q)}|_D$ is still Hermitian and satisfies  (\ref{eqfixA}) for $\psi \in D$. This is just the  situation  treated in (a) and (b) of 
Theorem \ref{Theorem: uniqueness of POVM associated with maximally symmetric operator} when identifying  $A=A^{(Q^{(A)})}|_{D(A)}$  where $Q^{(A)}$ is a 
POVM decomposing the symmetric operator  $A$ according to (a) Theorem \ref{Theorem: uniqueness of POVM associated with maximally 
symmetric operator}. Here $D(A)\subset D(A^{(Q^{(A)})})$ is dense and, in general, $A^{(Q^{(A)})} \supset A$ and (\ref{v234}) is valid on $D(A)$, but not on $D(A^{(Q^{(A)})})$. If $A$
is maximally symmetric, $A=A^{(Q^{(A)})}$ necessarily.\\
{\bf (2)} Theorems \ref{Theorem: uniqueness of POVM associated with maximally symmetric operator} and \ref{TheoremPOVMtoA} can be used to define a function $f(A)$ of a symmetric  operator $A$ in $\sH$ when $A$ itself can be decomposed along the normalized POVM $Q^{(A)}$ according to Definition \ref{DEFassociate}, and  $f: \bR \to \bR$ is Borel measurable. 
It is simply sufficient to observe that $Q'(E) := Q^{(A)}(f^{-1}(E))$ is still a normalized POVM when $E$ varies in $\cB(\bR)$, so that $f(A)$ can be defined  according to  definitions (\ref{defD2}) and
 (\ref{eqfixA}) just by integrating $Q'$.  Therefore, from the standard measure theory, it arises
\begin{align}\label{int23}
&\langle \varphi| f(A) \psi \rangle =   \int_\bR \mu \:dQ'_{\varphi,\psi}(\mu) = \int_\bR f(\lambda) \:dQ_{\varphi,\psi}^{(A)}(\lambda)\quad \mbox{if $\varphi\in \sH$ and
$\psi \in D(f(A))$\:,}\\
&D(f(A)) = 
 \left\{\psi \in \sH \:\left|\:\int_{\mathbb{R}}\mu^2 dQ'_{\psi,\psi}(\mu)<+\infty\right. \right\} =\left\{\psi \in \sH \:\left|\:\int_{\mathbb{R}} |f(\lambda)|^2 dQ^{(A)}_{\psi,\psi}(\lambda)<+\infty\right. \right\}\:.
\end{align}
When $A$ is selfadjoint, so that we  deal with a PVM, this definition of $f(A)$ coincides to the standard one. It is however necessary to stress that, when $Q$ is properly  a POVM,
\begin{itemize}
\item[(a)]  unless $f$ is bounded (in that case  $D(f(A))=\sH$),  there is no guarantee  that the Hermitian operator $f(A)$ has a dense domain nor that $||f(A)\psi||^2 =
 \int_{\mathbb{R}} |f(\lambda)|^2 dQ^{(A)}_{\psi,\psi}(\lambda)$
for $\psi \in D(f(A))$ as observed in  Remark (1)\footnote{In particular, if 
$f: \bR \ni \lambda \to \lambda \in \bR$, we have $A=f(A)\sp\rest_{D(A)}$, but the domain of $f(A)$ according to  (\ref{defD2}) is in general larger than $D(A)$, and
 $||f(A)\psi||^2 = \int_{\mathbb{R}} |f(\lambda)|^2 dQ^{(A)}_{\psi,\psi}(\lambda)$ is valid for $\psi \in D(A)$.},

\item[(b)] $f(A)$ does not satisfy the same properties as those of the standard functional calculus of  selfadjoint operators ,  just because the fundamental property of PVMs $Q(E)Q(E') = Q(E\cap E')$ is  false for POVMs, 
\item[(c)] the notion of $f(A)$ also depends on the normalized POVM $Q^{(A)}$ exploited to decompose $A$, since $Q^{(A)}$ is unique if and only if $A$ is maximally symmetric.\hfill $\blacksquare$
\end{itemize}}
\end{remark}

\section{Observable interpretation of $\pi_\omega(a)$ in terms of POVMs and expectation-value interpretation of $\omega(a)$}\label{Sec: Generalized observable pi(a) and expectation-value interpretation of omega(a)}
We are in a position to apply the developed theory to tackle the initial problems stated in issues A and B establishing the main results of this work.
\subsection{Interpretation of $\pi_\omega(a)$  in terms of POVMs}
Coming back to symmetric operators arising from GNS representations, the summarized  theory of POVMs  and   Corollary \ref{corollaryessselfadj} in particular have some important consequences concerning a possible interpretation of $\pi_\omega(a)$
as a generalized observable when it is not essentially selfadjoint. Consider  
the symmetric operator $\pi_\omega(a)$ when  $a^*=a \in \gA$ and $\omega: \gA \to \bC$ is a non-normalized state on the unital $*$-algebra $\gA$. We have that

\begin{itemize}
\item[(1)] $\pi_\omega(a)$ and $\overline{\pi_\omega(a)}$ share the same class of
 associated normalized POVMs $Q^{(a,\omega)}$ so that they support the same physical information 
when interpreting them as generalized observables. More precisely,  each of these POVMs endows those symmetric 
operators with the physical meaning of generalized observable in the Hilbert space $\sH_\omega$. This is particularly relevant
 when $\pi_\omega(a)$ does not admit selfadjoint extensions;
\item[(2)]  the above class of normalized POVMs however  includes also all possible PVMs of all possible selfadjoint extensions of $\pi_\omega(a)$  provided they exist. Hence, the standard notion of quantum observable in Hilbert space is encompassed;
\item[(3)]   $Q^{(a,\omega)}$ is unique
 if and only if 
$\overline{\pi_\omega(a)}$ is maximally symmetric but not necessarily selfadjoint;
\item[(4)]   That  unique POVM is a PVM if $\pi_\omega(a)$ 
is  essentially selfadjoint.
\end{itemize}
Even if  the symmetric operator $\pi_\omega(a)$ does {\em not} admit a selfadjoint extension, it can be considered a generalized  observable with some precautions, since it admits decompositions in terms of POVMs which are {\em generalized observables} in their own right. However, in general, there are {\em many} POVMs associated with {\em one} given symmetric operator $\pi_\omega(a)$. Next section tackles the problem of reducing this number in relation with the expectation-value interpretation of $\omega_b(a)$.

\subsection{Expectation-value interpretation of $\omega_b(a)$  by means of consistent class of measures solving the moment problem}
Let us now come to the  expectation-value interpretation of $\omega(a)$  extended to the deformations $\omega_b(a)$. This
 inteprentation relies upon  the choice of a
 measure $\mu^{(a)}_{\omega}$ viewed as a particular case of the large class of measures $\mu^{(a)}_{\omega_b}$
 associated to deformed states $\omega_b$. All these measures are assumed to 
 solve the  moment problem (\ref{bmoment}) for $(a,\omega_b)$, where the case  $n=1$ is just the expectation-value interpretation of $\omega_b(a)$ and $\omega(a)$ in particular for $b=\bI$.

The final discussion  in (1) and (2) in Example \ref{Ex: uniqueness of moment problem and self-adjointness of the operator are not related}  shows that, in physically relevant cases,  there are many measures solving the moment problem
 relative to $(a,\omega_b)$ in general, even if the operator $\pi_\omega(a)$ is essentially selfadjoint.  
We need some physically meaningful strategy to reduce the number of those measures.

This section proves that, once we have imposed suitable physically meaningful requirements on the measures $\mu^{(a)}_{\omega_b}$, 
a new connection arises between the remaining classes of physically meaningful measures and POVMs decomposing the symmetric
 operators $\pi_{\omega_b}(a)$. 
These POVMs also generate the said  measures $\mu^{(a)}_{\omega_b}$.

We start by noticing that when $b$ is a  function of the Hermitian element $a\in \gA$,
	the measures $\mu^{(a)}_{\omega_{b}}$ and $\mu^{(a)}_\omega$ are not independent, in particular it holds
	\begin{align*}
		\int_{\mathbb{R}}\lambda^{2k+1}d\mu^{(a)}_\omega(\lambda)=
		\omega(a^{2k+1})=\omega_{a^{k}}(a)=
		\int_{\mathbb{R}}\lambda\,d\mu^{(a)}_{\omega_{a^{k}}}(\lambda)\,.
	\end{align*} 
 However, referring only to the subalgebra generated by $a$, we miss the information of the whole algebra $\gA$ which contains $a$.  
We therefore try to  restrict the class  of the  measures $\mu^{(a)}_{\omega_b}$ by imposing some  natural compatibility conditions 
 among the measures $\mu_{\omega_b}^{(a)}$ associated with {\em completely general} elements  $b\in \gA$.
As a starting point, let us consider a family of measures $\{\mu^{(a)}_{\omega_b}\}_{b \in\mathfrak{A}}$ each of which is a solution to the moment problem (\ref{bmoment}) relative to $(a,\omega_b)$, with $\mu^{(a)}_{\omega_b}=0$ if $\omega_b$ is singular.
Since for all $b,c\in\mathfrak{A}$,  $z\in\mathbb{C}$, and every real polynomial $p$,
\begin{align}  \nonumber  
	\omega_{b+c}(p(a))+\omega_{b-c}(p(a)) &= 2 [\omega_{b}(p(a))+\omega_{c}(p(a)) ]\,.\\
	\omega_{zb}(p(a))&=|z|^2\omega_{b}(p(a)) \,,
\end{align} 
we also have
\begin{align}\label{Equation: parallelogram rule for moments}
	\int_{\mathbb{R}}p(\lambda)d\mu^{(a)}_{\omega_{b+c}}(\lambda)+
	\int_{\mathbb{R}}p(\lambda)d\mu^{(a)}_{\omega_{b-c}}(\lambda)&=2\bigg[
	\int_{\mathbb{R}}p(\lambda)d\mu^{(a)}_{\omega_{b}}(\lambda)+
	\int_{\mathbb{R}}p(\lambda)d\mu^{(a)}_{\omega_{c}}(\lambda)
	\bigg]\,.\\
	\label{Equation: quadratic rule for moments}
	\int_{\mathbb{R}}p(\lambda)d\mu^{(a)}_{\omega_{z b}}(\lambda)&=|z|^2\int_{\mathbb{R}}p(\lambda)d\mu^{(a)}_{\omega_b}(\lambda)\,,
\end{align}
Finally observe that the following  directional  continuity property holds true for $b, c, a=a^* \in \mathfrak{A}$, and every real polynomial $p$,
\begin{align}\nonumber
\omega_{b+ t c}(p(a)) \to \omega_{b}(p(a))\quad \mbox{for $\mathbb{R} \ni t \to 0$}\,,
\end{align}
that implies
\begin{align}\label{Equation:directional-continuity}
\int_{\mathbb{R}}p(\lambda)d\mu^{(a)}_{\omega_{b+tc}}(\lambda) \to 
\int_{\mathbb{R}}p(\lambda)d\mu^{(a)}_{\omega_{b}}(\lambda)
\quad \mbox{for $\mathbb{R} \ni t \to 0$\,.}
\end{align}

\noindent Identities 
 (\ref{Equation: parallelogram rule for moments})-(\ref{Equation:directional-continuity}) are true for every choice of measures associated with the algebraic observable $a$
and the deformations $\omega_b$, so that they {\em cannot} be used as contraints to reduce the number of those measures.

 We observe that the above relations actually regard {\em polynomials} $p(a)$ of $a$. From the physical side, dealing only with polynomials  seems a limitation since   we  expect that, at the end of the game, after having introduced some technical information,   one would be able to define more complicated functions of $a$ (as it happens when dealing with $C^*$-algebras), because these observables are 
physically necessary and have a straightforward operational definition: {\em $f(a)$ heuristically  represents the algebraic observable (if any) that attains the values $f(\lambda)$, where $\lambda$ are the values attained by $a$}. 
If $\mu^{(a)}_{\omega_b}$ is physically meaningful and $f: \bR \to \bR$ is a bounded  measurable function (we restrict ourselves to bounded functions to avoid subtleties with domains),  we expect that the (unknown) algebraic observable  $f(a)$ is however  represented by the function $f(\lambda)$
in the space $L^2(\bR, d\mu^{(a)}_{\omega_b})$ and, as far as expectation values are concerned, $\omega_b(f(a)) = \int_\bR f(\lambda)\: d \mu_{\omega_b}^{(a)}(\lambda)$.

This viewpoint can be also  heuristically supported from another side. 
If we deal with $\pi_{\omega_b}(a)$ instead of $a$ itself and we decompose the symmetric operator $\pi_{\omega_b}(a)$ with a POVM, the function $f(\pi_{\omega_b}(a))$ can be defined according to (2) Remark \ref{remfA}. If we now  assume that $\mu^{(a)}_{\omega_b} = Q^{(\pi_{\omega_b}(a))}_{\psi_{\omega_b},\psi_{\omega_b}}$ we just have that 
 the algebraic observable  $f(a)$ is  represented by the function $f(\lambda)$ when we compute the expectation values: according to (\ref{int23}) for $\varphi=\psi=\psi_{\omega_b}$, we have  $\omega_b(f(a)) = \langle \psi_{\omega_b}| \pi_{\omega_b}(f(a))\psi_{\omega_b} \rangle =\int_\bR f(\lambda)\: d \mu_{\omega_b}^{(a)}(\lambda)$.

We therefore  strengthen equations (\ref{Equation: parallelogram rule for moments})-(\ref{Equation:directional-continuity})  {\em by requiring that the 
physically interesting measures are such that (\ref{Equation: parallelogram rule for moments})-(\ref{Equation:directional-continuity}) are valid 
for arbitrary bounded measurable  functions $f: \mathbb{R} \to  \mathbb{R}$ in place of polynomials $p$. }

The resulting condition, just specializing to  characteristic functions $f=\chi_E$  for every Borel measurable set $E$ over the real line,  leads to the following identities, which imply the previous ones (stated for general bounded measurable functions)
\begin{align}\label{A}
	\mu^{(a)}_{\omega_{b+c}}+\mu^{(a)}_{\omega_{b-c}}=2\big[\mu^{(a)}_{\omega_{b}}+\mu^{(a)}_{\omega_{c}}\big]\,,\quad
	\mu^{(a)}_{\omega_{zb}}=|z|^2\mu^{(a)}_{\omega_{b}}\\
	\label{B}
	\mu^{(a)}_{\omega_{b+tc}}(E) \to \mu^{(a)}_{\omega_{b}}(E)\quad \mbox{if $\mathbb{R}\ni t\to 0$,} 
\end{align}

\begin{remark}
{\em {\bf (1)} We stress that (\ref{A}) and (\ref{B})  are {\em not} consequences of  (\ref{Equation: parallelogram rule for moments})-(\ref{Equation:directional-continuity}) in the general case, in particular because polynomials are not necessarily dense in the relevant $L^1$ spaces, since the considered Borel measures have non-compact support in general and we cannot directly apply Stone-Weierstass theorem.
 (\ref{A}) and (\ref{B}) are  however necessarily satisfied when all considered measures $\mu_{\omega_b}^{(a)}$ are induced by a unique PVM as for instance  in the strong hypotheses of Theorem \ref{T1}:
(\ref{muPVMb}) immediately implies
(\ref{A}) and (\ref{B}). This is also the case for a $C^*$-algebra, since the measures arise from PVMs due to Proposition \ref{P2}.\\
{\bf (2)} Identities (\ref{A}) and (\ref{B}) remain valid also  when labeling the measures with the classes $[b] \in \mathfrak{A}/G( \mathfrak{A}, \omega) =\mathcal{D}_\omega$
since these only involve the linear structure of $\gA$ which survive the quotient operation.} \hfill $\blacksquare$\\
\end{remark} 
\noindent We can state the following general definition, taking remark (2) into account in particular.

\begin{definition}\label{Definition: compatibility condition for family of measures}
 {\em If $\cD$ is a complex vector space, a family of positive $\sigma$-additive measures  $\{\nu_{\psi}\}_{\psi\in \cD}$ over the measurable space $(\Omega, \Sigma)$  such that
	\begin{align}\label{Equation: compatibility conditions for measures}
	\nu_{\psi+\varphi}+\nu_{\psi-\varphi}=2\big[\nu_{\psi}+\nu_{\varphi}\big]\,,\quad
	\nu_{z \psi}=|z|^2\nu_\psi\,\qquad \mbox{for all $\psi,\varphi\in\cD$ and $z\in\mathbb{C}$}\\
	\label{Equation: compatibility conditions for measures2}
	\nu_{\psi+t \varphi}(E) \to \nu_{\psi}(E) \quad \mbox{if $\mathbb{R}\ni t\to 0$,} \quad \mbox{for every fixed triple  $\psi,\varphi\in\cD$ and $E\in\Sigma$.}
	\end{align}
 is said to be {\bf consistent}.}\hfill $\blacksquare$\\
\end{definition} 

\begin{remark}
	{\em From  Definition \ref{Definition: compatibility condition for family of measures}, $\nu_0$ is the zero measure ($\nu_0(E)=0$ if $E\in \Sigma$).}\hfill $\blacksquare$
\end{remark}
%

\subsection{Consistent classes of measures and POVMs}
We now apply the summarized theory of POVMs to prove that the family of  POVMs associated to $\pi_{\omega}(a)$ is one-to-one with the family of consistent classes of measures  solving the moment problem for all $\omega_b$.
The proof consists of two steps. Here is the former.

If $Q^{(a,\omega)}$ is a POVM associated to $\pi_\omega(a)$ for $a^*=a \in \gA$ and for a non-normalized state $\omega : \gA \to \bC$,
 let $\nu^{(a)}_{\omega_b}$ be the Borel  measure defined by
	\begin{align}\label{muQ}
	\nu^{(a)}_{\omega_b}(E):=\langle\psi_{\omega_b}|Q^{(a,\omega)}(E)\psi_{\omega_b}\rangle \quad \mbox{ if $E\in \cB(\bR)$}\,,
	\end{align}
for every deformation $\omega_b$.

\begin{theorem} \label{theorem-a-POVM} Consider  the unital $*$-algebra  $\mathfrak{A}$,  a non-normalized state $\omega: \mathfrak{A} \to \mathbb{C}$, an element $a=a^* \in  \mathfrak{A}$ and the family of measures
$\{\nu^{(a)}_{\omega_b}\}_{b \in \mathfrak{A}}$ defined in (\ref{muQ}) with respect to a  normalized POVM $Q^{(a,\omega)}$  associated to 
 $\pi_\omega(a)$.
Then

{\bf (a)}  $\{\nu^{(a)}_{\omega_b}\}_{b \in \mathfrak{A}}$ is  a consistent  family over $\cD_\omega = \gA/G(\gA,\omega)$.

{\bf (b)} Each $\nu^{(a)}_{\omega_b}$ is a  solution of the moment problem (\ref{bmoment}) relative to 
$(a, \omega_b)$.
\end{theorem}

\begin{proof}
Let us focus on Theorem  \ref{Theorem: uniqueness of POVM associated with maximally symmetric operator} for 
$A:= \pi_\omega(a)$ with $D(A)= \mathcal{D}_\omega$  and $\mathsf{H}= \mathsf{H}_\omega$.
According to (b), the POVM $Q^{(A)}= Q^{(a,\omega)}$ can be written as $Q^{(A)} = P_{\mathsf{H}}P|_{\mathsf{H}}$ for a PVM $P$ of a   selfadjoint operator $B: D(B) \to \mathsf{K}$
defined on a larger Hilbert space $\mathsf{K}$, including $\mathsf{H}$ as a closed subspace, such that $A = B\spa\rest_{D(A)}$.
Observe that $\pi_{\omega}(b)\psi_\omega \in \mathcal{D}_\omega= D(\pi_\omega(a^n)) = D(\pi_\omega(a)^n)= D(A^n) \subset D(B^n)$ where, in the last inclusion, we have exploited $A=B\spa\rest_{D(A)}$
and $A( \mathcal{D}_\omega)\subset  \mathcal{D}_\omega$.
 By the standard spectral theory of selfadjoint operators  we therefore have ($\psi_{\omega_b}:= \pi_\omega(b)\psi_\omega$)
\begin{align*}
	\langle \pi_{\omega}(b)\psi_\omega| B^n \pi_{\omega}(b)\psi_\omega \rangle =
	\int_{\mathbb{R}}\lambda^n dP_{\psi_{\omega_b},\psi_{\omega_b}}(\lambda) =
	\int_{\mathbb{R}}\lambda^n dQ^{(A)}_{\psi_{\omega_b},\psi_{\omega_b}}(\lambda)=
	\int_{\mathbb{R}}\lambda^n d\nu^{(a)}_{\omega_b}(\lambda)\,,
\end{align*}
where, in the last passage we have used  $Q^{(A)} = P_{\mathsf{H}}P|_{\mathsf{H}}$, $ P_{\mathsf{H}}\pi_{\omega}(b)\psi_\omega  = \pi_{\omega}(b)\psi_\omega $, and (\ref{muQ}). On the other hand, per construction, 
$A^n \pi_{\omega}(b)\psi_\omega = B^n \pi_{\omega}(b)\psi_\omega$ and eventually
 the GNS theorem yields
$ \omega_b(a^n) = \langle \pi_{\omega}(b)\psi_\omega| A^n \pi_{\omega}(b)\psi_\omega \rangle = \langle \pi_{\omega}(b)\psi_\omega| B^n \pi_{\omega}(b)\psi_\omega \rangle$.
In summary, if $n=0,1,2\ldots$ and $b\in \mathfrak{A}$,
$$ \omega_b(a^n)=  \int_{\mathbb{R}}\lambda^n dQ^{(A)}_{\psi_{\omega_b},\psi_{\omega_b}}(\lambda)\:.$$
We have established that each measure (\ref{muQ})  is a solution of the moment problem relative to
$(a, \omega_b)$.
	By direct inspection, one immediately sees that  $\{\nu^{(a)}_{\omega_b}\}_{b\in \mathfrak{A}}$ defined as in (\ref{muQ}) satisfies  Definition \ref{Definition: compatibility condition for family of measures}.
\end{proof}

The result is reversed  with the help of the following abstract  technical proposition.

\begin{proposition}\label{technicalproposition} Let $X$  be a complex vector space and $p: X \to [0,+\infty)$ 
 such that
\begin{itemize}
\item[(i)] $p(\lambda x) = |\lambda| p(x)$ for every pair  $x\in X$ and $\lambda \in \mathbb{C}$,
\item[(ii)] $p(x+y)^2 + p(x-y)^2 = 2[p(x)^2 + p(y)^2]$ for every pair $x,y \in X$,
\item[(iii)] $p(x+ty) \to p(x)$ for $\mathbb{R} \ni t \to 0^+$ and every fixed pair $x,y \in X$.
\end{itemize}
Under these hypotheses, 
\begin{itemize}
\item[{\bf (a)}] $p$ is a seminorm on $X$,

\item[{\bf (b)}]  there is a unique positive semi definite Hermitian scalar product  $X\times X \ni (x,y) \mapsto (x|y)_p \in \mathbb{C}$ such that 
$p(x) = \sqrt{(x|x)_p}$ for all  $x\in X$,
\item[{\bf (c)}] the scalar product in (b) satisfies
\begin{equation} (x|y)_p = \frac{1}{4} \sum_{k=0}^3 (-i)^k p(x+i^ky)^2\qquad \mbox{for $x,y \in X$.} \label{scalarprodfromp}\end{equation}
\end{itemize}
\end{proposition}

\begin{proof}
See Appendix \ref{AppendixB}
\end{proof}

\noindent We can now establish another main result of the work,  which is the converse of Theorem \ref{theorem-a-POVM}. Together with the afore-mentioned theorem, it proves that 
for $a^*=a\in \gA$ and a non-normalized state $\omega: \gA \to \bC$,
the family of normalized POVMs associated with the given symmetric operator $\pi_\omega(a)$ is  one-to-one with the family of consistent classes of measures  of all $\omega_b(a)$ which solve the moment problem for all deformations $\omega_b$,
when $b\in \gA$.

\begin{theorem}\label{Theorem: characterization of compatible measures}
	Consider  the unital $*$-algebra  $\mathfrak{A}$,  a non-normalized state $\omega: \mathfrak{A} \to \mathbb{C}$,
$a=a^* \in  \mathfrak{A}$, and   a consistent class of measures  $\{\mu^{(a)}_{\omega_b}\}_{b \in \mathfrak{A}}$
 solutions of the moment problem relative to the pairs  $(a, \omega_b)$ for $b \in \mathfrak{A}$.  Then

{\bf (a)}	There is a unique normalized POVM $Q^{(a, \omega)}\colon\mathscr{B}(\mathbb{R})\to\mathcal{B}(\mathsf{H}_\omega)$ such that, if $b\in\mathfrak{A}$,
	\begin{align}\label{POVMmu}
		\mu^{(a)}_{\omega_b}(E) =\langle \psi_{\omega_b}|Q^{(a, \omega)}(E)\psi_{\omega_b}\rangle\qquad\forall E\in\mathscr{B}(\mathbb{R})\,,
	\end{align}

	{\bf (b)}	  $Q^{(a, \omega)}$ decomposes  $\pi_\omega(a)$ according to Definition \ref{DEFassociate} so  that,
in particular,
\begin{equation}\label{incl1} \mathcal{D}_\omega =D(\pi_\omega(a)) 
\subset \left\{ \psi \in \mathsf{H}_\omega \:\left|\:  \int_{\mathbb{R}} \lambda^2 
d Q^{(a, \omega)}_{\psi,\psi}(\lambda) < +\infty \right.\right\}\:.\end{equation}

{\bf (c)}	 $Q^{(a, \omega)}$  decomposing $\pi_\omega(a)$  is unique
 if and only if $\overline{\pi_\omega(a)}$ is maximally symmetric. In this case
\begin{equation}\label{incl2}  D(\overline{\pi_\omega(a)})  =  \left\{ \psi \in \mathsf{H}_\omega \:\left|\:  \int_{\mathbb{R}} \lambda^2 
d Q^{(a,\omega)}_{\psi,\psi}(\lambda) < +\infty \right.\right\}\end{equation}
and that unique $Q^{(a,\omega)}$ is a PVM if and only if  $\overline{\pi_\omega(a)}$  is selfadjoint. In that case  $Q^{(a,\omega)}$ coincides with the PVM of 
$\overline{\pi_\omega(a)}$.

{\bf (d)} Let us define
\beq \label{newQ}Q^{(a,\omega_b)}(E) := P_{\omega_b}Q^{(a,\omega)}(E)\sp\rest_{\sH_{\omega_b}}\:,\eeq where $P_{\omega_b}: \sH_\omega \to \sH_\omega$ is the orthogonal projector onto 
$\sH_{\omega_b}$. It turns out that, for $b\in \gA$,
\begin{itemize}
\item[(i)]  $Q^{(a,\omega_b)}$ is a normalized POVM in $\sH_{\omega_b}$.
 \item[(ii)] It holds
\begin{align}\label{POVMmu2}
		\langle \psi_{\omega_b}|Q^{(a, \omega_b)}(E)\psi_{\omega_b}\rangle =\mu^{(a)}_{\omega_b}(E) \qquad\forall E\in\mathscr{B}(\mathbb{R})\:,
	\end{align}
\item[(iii)] $Q^{(a, \omega_b)}$ decomposes $\pi_{\omega_b}(a)$ in the sense of Definition \ref{DEFassociate}.
\end{itemize}
\end{theorem}
\begin{proof}

	(a) It is clear that, if a normalized POVM exists satisfying (\ref{POVMmu}) for all $b\in\mathfrak{A}$, then  it is unique.  Indeed, taking advantage of the 
polarization identity, another similar POVM $Q$ would satisfy $\langle \psi_{\omega_b}| (Q(E)-Q^{(a,\omega)}(E)) \psi_{\omega_c} \rangle=0$ for every 
$\psi_{\omega_b},\psi_{\omega_c} \in \mathcal{D}_\omega$, which is  a dense set. Therefore, $(Q(E)-Q^{(a,\omega)}(E)) \psi_{\omega_c}=0$  for every 
$\psi_{\omega_c} \in \mathcal{D}_\omega$. Continuity of $Q(E)-Q^{(a,\omega)}(E)$ yields $Q(E)=Q^{(a,\omega)}(E)$. \\
Let us prove that a normalized POVM satisfying  (\ref{POVMmu}) for all $b\in\mathfrak{A}$ exists.
 Fix  $E\in\mathscr{B}(\mathbb{R})$.
Since the positive measures $\nu_{\psi_{\omega_b}} := \mu^{(a)}_{\omega_b}$ satisfy the identities (\ref{Equation: compatibility conditions for measures}) and (\ref{Equation: compatibility conditions for measures2}),
$\mathcal{D}_\omega \ni \psi \mapsto \sqrt{\nu_\psi(E)}$ fulfills the hypotheses of Proposition \ref{technicalproposition}. Consequently, that function is a seminorm over $\mathcal{D}_\omega$ and there is a unique semidefinite Hermitian scalar product  
inducing it:
	\begin{align}\label{Equation: definition of scalar product associated with compatible measures}
		(\psi_{\omega_b}|\psi_{\omega_c})_E:=\frac{1}{4}\sum_{k=0}^3(-i)^k\mu^{(a)}_{\omega_{b+i^kc}}(E)\,.
	\end{align}
	Applying Cauchy-Schwarz inequality, we have
	\begin{align*}
		|(\psi_{\omega_b}|\psi_{\omega_c})_E|^2\leq\|b\|^2_E\|c\|_E^2=
		\mu^{(a)}_{\omega_b}(E)\mu^{(a)}_{\omega_c}(E)\leq
		\mu^{(a)}_{\omega_b}(\mathbb{R})\mu^{(a)}_{\omega_c}(\mathbb{R})=
		\omega_b(\mathbb{I}) \omega_c(\mathbb{I}) = ||\psi_{\omega_b}||_{\mathsf{H}_\omega}^2\: ||\psi_{\omega_c}||_{\mathsf{H}_\omega}^2\,.
	\end{align*}
	Exploiting Riesz' theorem, it then follows that $(\,|\,)_E$ continuously extends to $\mathsf{H}_\omega\times\mathsf{H}_\omega$ and
	moreover there exists a unique  selfadjoint  positive  operator $Q(E)\in\mathfrak{B}(\mathsf{H}_\omega)$ with $0\leq Q(E)\leq 1$ such that
	\begin{align}
		(\psi|\varphi)_E=\langle\psi|Q(E)\varphi\rangle\qquad\forall\psi,\varphi\in\mathsf{H}_\omega\,.
	\end{align}
	The map $Q\colon\mathscr{B}(\mathbb{R})\ni E\mapsto Q(E)\in\mathfrak{B}(\mathsf{H}_\omega)$ is a normalized POVM according to Definition \ref{Definition: POVM} as we go to prove. In fact, 
$Q(E) \geq 0$ as said above and,
since 
\begin{align}\label{sigma}
		\langle\psi_{\omega_b}|Q(E)\psi_{\omega_c}\rangle = \frac{1}{4}\sum_{k=0}^3(-i)^k\mu^{(a)}_{\omega_{b+i^kc}}(E) \qquad\forall\psi_{\omega_b},\psi_{\omega_c}\in\mathcal{D}_\omega\,,
	\end{align}
the left-hand side is a complex Borel measure over $\mathbb{R}$ since the right-hand side is a complex combination of such measures. Finally, $E= \mathbb{R}$ produces
\begin{align}
		\langle\psi_{\omega_b}|Q(\mathbb{R})\psi_{\omega_c}\rangle = \frac{1}{4}\sum_{k=0}^3(-i)^k ||\psi_{\omega_b} +i^k \psi_{\omega_c}||^2
= \langle\psi_{\omega_b}|\psi_{\omega_c}\rangle \qquad\forall\psi_{\omega_b},\psi_{\omega_c}\in\mathcal{D}_\omega\,.
	\end{align}
As $\mathcal{D}_\omega$ is dense in $\mathsf{H}_\omega$, it implies $Q(\mathbb{R})=I$ so that the candidate POVM  $Q$ is normalised. 
To conclude the proof of the fact that $Q^{(a,\omega)}:= Q$ is a POVM, it is sufficient to prove that $\mathscr{B}(\mathbb{R})\ni E \mapsto \langle\psi|Q(E)\varphi\rangle \in \mathbb{C}$ is a complex measure no matter we choose  $\psi, \phi \in \mathsf{H}_\omega$ (and not only for  $\psi, \phi \in \mathcal{D}_\omega$ as we already know).
A continuity argument from the case of $\psi, \varphi \in \mathcal{D}_\omega$ proves that the said map is at least additive so that, in particular
$\langle\psi|Q(\emptyset)\varphi\rangle=0$, because $Q(\mathbb{R})=I$. Let us pass to prove that the considered function is unconditionally $\sigma$-additive so that it is a complex measure as wanted. If the sets $E_n \in \mathscr{B}(\mathbb{R})$ when $n\in \mathbb{N}$ satisfy $E_k \cap E_h = \emptyset$ for $h\neq k$, consider the difference 
$$\Delta_N := \sum_{n=0}^N \langle \psi| Q(E_n) \varphi \rangle - \langle \psi| Q(E) \varphi \rangle $$
where $E:= \cup_{n \in \mathbb{N}} E_n$. We want to prove that $\Delta_N \to 0$ for $N\to +\infty$.  $\Delta_N$ 
can be decomposed as follows
\begin{align}
	\Delta_N &=
	\sum_{n=0}^N \langle \psi -\psi_{\omega_b}| Q(E_n) (\varphi-\psi_{\omega_c}) \rangle +
	\sum_{n=0}^N \langle \psi_{\omega_b}| Q(E_n) (\varphi-\psi_{\omega_c}) \rangle +
	\sum_{n=0}^N \langle \psi-\psi_{\omega_b}| Q(E_n) \psi_{\omega_c} \rangle
	\nonumber \\
	&+\sum_{n=0}^N \langle \psi_{\omega_b}| Q(E_n) \psi_{\omega_c} \rangle -
	\langle \psi_{\omega_b}| Q(E) \psi_{\omega_c} \rangle\nonumber \\
	&- \langle \psi -\psi_{\omega_b}| Q(E) (\varphi-\psi_{\omega_c}) \rangle -
	\langle \psi_{\omega_b}| Q(E) (\varphi-\psi_{\omega_c}) \rangle -
	\langle \psi-\psi_{\omega_b}| Q(E) \psi_{\omega_c} \rangle \:. \nonumber 
\end{align}
 Using additivity and defining $F_N := \cup_{n=0}^N E_n$, we can re-arrange the found expansion as 
\begin{align}
	\Delta_N &=
	\langle \psi -\psi_{\omega_b}| Q(F_N) (\varphi-\psi_{\omega_c}) \rangle +
	\ \langle \psi_{\omega_b}| Q(F_N) (\varphi-\psi_{\omega_c}) \rangle +  
	\langle \psi-
	\psi_{\omega_b}| Q(F_N) \psi_{\omega_c} \rangle \nonumber \\
	&+\sum_{n=0}^N \langle \psi_{\omega_b}| Q(E_n) \psi_{\omega_c} \rangle -
	\langle \psi_{\omega_b}| Q(E) \psi_{\omega_c} \rangle\nonumber \\
	&- \langle \psi -\psi_{\omega_b}| Q(E) (\varphi-\psi_{\omega_c}) \rangle -
	\langle \psi_{\omega_b}| Q(E) (\varphi-\psi_{\omega_c}) \rangle -
	\langle \psi-\psi_{\omega_b}| Q(E)
	 \psi_{\omega_c} \rangle \:.\nonumber
\end{align}
Since $||Q(E)|| \:,  ||Q(F_N)|| \leq ||Q(\mathbb{R})||=1$, we have the estimate
\begin{align*}
	|\Delta_N| &\leq 2||\psi -
	\psi_{\omega_b}|| ||\varphi -\psi_{\omega_c}|| +
	2||\psi_{\omega_b}|| ||\varphi -\psi_{\omega_c}|| +
	2||\psi -\psi_{\omega_b}|| ||\psi_{\omega_c}||\\&+
	\left| \sum_{n=0}^N \langle \psi_{\omega_b}| Q(E_n) \psi_{\omega_c} \rangle - \langle \psi_{\omega_b}| Q(E) \psi_{\omega_c} \rangle\right|\:.
\end{align*}
This inequality concludes the proof:  given $\psi,\varphi \in \mathsf{H}_\omega$, since $\mathcal{D}_\omega$ is dense,  we can fix
 $\psi_{\omega_b},\psi_{\omega_c} \in \mathcal{D}_\omega$ such that the sum of the first three addends is bounded by $\epsilon/2$. Finally, 
exploiting the fact that $\mathscr{B}(\mathbb{R}) \ni E \mapsto \langle\psi_{\omega_b}|Q^{(a)}(E)\psi_{\omega_c}\rangle$ is  $\sigma$-additive, 
we can fix $N$ sufficiently large that the last addend is bounded by $\epsilon/2$. So, if $\epsilon>0$, there
 is $N_\epsilon $ such that $|\Delta_N|  < \epsilon $ if $N > N_\epsilon$ as wanted. Notice that  the series $\sum_{n=0}^{+\infty} \langle \psi| Q(E_n) \varphi \rangle$ can be re-ordered arbitrarily since we have proved that its sum is $\langle \psi| Q(E) \varphi \rangle$ 
which does not depend on the order used to label the sets $E_n$ because  $E:= \cup_{n \in \mathbb{N}} E_n$. 
The function 
$\mathscr{B}(\mathbb{R}) \ni B \mapsto \langle \psi| Q(B) \varphi \rangle \in \mathbb{C}$ is {\em unconditionally} $\sigma$-additive as we wanted to prove.\\
(b)  Since the measures $\mu_{\omega_b}^{(a)}$ are solutions of the moment problem for the pairs $(a, \omega_b)$, for $k=0,1,2, \ldots$ and $\psi_{\omega_b} \in \mathcal{D}_\omega$, we have
\begin{equation}\label{chain}\langle \psi_{\omega_b} | \pi_\omega(a)^k \psi_{\omega_b} \rangle = \omega_b(a^k) = \int_{\mathbb{R}} \lambda^k d \mu_{\omega_b}^{(a)}(\lambda) = 
 \int_{\mathbb{R}} \lambda^k d Q^{(a,\omega)}_{\psi_{\omega_b},\psi_{\omega_b}}(\lambda)\:.\end{equation}
Choosing $k=2$ we obtain 
$$||\pi_\omega(a) \psi||^2 =  \int_{\mathbb{R}} \lambda^2 d Q^{(a,\omega)}_{\psi,\psi}(\lambda) \quad \mbox{for every $\psi \in D(\pi_\omega(a)) = \mathcal{D}_\omega$,}$$
which, in particular, also implies (\ref{incl1}).  It remains to be established the identity
\begin{equation}\label{last}\langle \varphi | \pi_\omega(a) \psi \rangle = 
 \int_{\mathbb{R}} \lambda d Q^{(a,\omega)}_{\varphi,\psi}(\lambda)\quad \mbox{for every $\varphi \in \mathsf{H}_\omega$ and $\psi \in D(\pi_\omega(a))$}\:.\end{equation}
From (\ref{chain}) with $k=1$ we conclude that, for every $\psi_{\omega_b} \in \mathcal{D}_\omega$,
$$\langle \psi_{\omega_b}| \pi_\omega(a) \psi_{\omega_b}\rangle =  \int_{\mathbb{R}} \lambda d Q^{(a, \omega)}_{\psi_{\omega_b},\psi_{\omega_b}}(\lambda) =
\langle \psi_{\omega_b}| A \psi_{\omega_b}\rangle \:,  $$
where $A$ is the Hermitian operator uniquely constructed out of the POVM $Q^{(a)}$ according to Theorem \ref{TheoremPOVMtoA}. Notice that the domain of $A$ is $\left\{ \psi \in \mathsf{H}_\omega \:\left|\:  \int_{\mathbb{R}} \lambda^2 
d Q^{(a,\omega)}_{\psi,\psi}(\lambda) < +\infty \right.\right\}$ that includes $\mathcal{D}_\omega$ for  (\ref{incl1}) that we have already proved.  Polarization identity applied to both sides of 
$\langle \psi_{\omega_b}| \pi_\omega(a) \psi_{\omega_b}\rangle =   \langle \psi_{\omega_b}| A \psi_{\omega_b}\rangle$
immediately proves that 
$\langle \varphi| \pi_\omega(a) \psi\rangle = \langle \varphi| A \psi\rangle $
for every pair $\varphi, \psi \in \mathcal{D}_\omega$. Since this space is dense, we conclude that $\pi_\omega(a) = A\spa\rest_{\mathcal{D}_\omega}$. Identity (\ref{last}) is therefore valid as an immediate consequence of (b) Theorem \ref{TheoremPOVMtoA}.\\
(c) Everything follows from (c), (d), (e), (f) of Theorem \ref{Theorem: uniqueness of POVM associated with maximally symmetric operator} and (a) of Corollary \ref{corollaryessselfadj}.\\
(d) (i) is true per direct inspection. (ii) is consequence of (\ref{POVMmu}) using $\psi_{\omega_b} \in \sH_{\omega_b} \subset \sH_\omega$. (iii) arises from the fact that
 $Q^{(a, \omega)}$ decomposes $\pi_{\omega}(a)$, (\ref{newQ}) and (\ref{GNSb}).
\end{proof}
\begin{remark}
	{\em Item (d)  physically states that the POVMs $Q^{(a,\omega_b)}$ are consistent with both the expectation-value  interpretation of each $\omega_b(a)$ and the interpretation of 
every $\pi_{\omega_b}(a)$
as generalized observable. We stress that, if $Q^{(a,\omega)}$ is a PVM because, for example, $\overline{\pi_\omega(a)}$ is selfadjoint, it is still possible that $Q^{(a, \omega_b)}$ is merely a POVM and not a PVM.}\hfill $\blacksquare$\\
\end{remark}
The next result can  be considered as a {\em weaker} version of both Theorem \ref{T1} and  its converse (the  proper converse  of Theorem \ref{T1} does not exist as we have seen).
\begin{corollary}\label{Cor: uniqueness of consistent measures}
	Let $\mathfrak{A}$ be a unital $*$-algebra,  $\omega\colon\mathfrak{A}\to\mathbb{C}$ a non-normalized state and $a=a^* \in \mathfrak{A}$.\\
$\overline{\pi_\omega(a)}$ is maximally symmetric if and only if there exists a unique family of consistent measures $\{\mu^{(a)}_{\omega_b}\}_{b\in\mathfrak{A}}$ solutions of the 
moment problem (\ref{bmoment}) relative to the pairs  $(a, \omega_b)$ for $b \in \mathfrak{A}$.
\end{corollary}
\begin{proof}
	It immediately follows from (c) of Theorem  \ref{Theorem: characterization of compatible measures}, equation (\ref{POVMmu}) and Theorem \ref{Theorem: uniqueness of POVM associated with maximally symmetric operator}.
\end{proof}

\begin{example} {\em Let us come back to the algebra $\gB$ equipped with the non-normalized state $\phi$ (\ref{statephi}) defined in (2) in Example \ref{Ex: algebraically self-adjoint observables which are not self-adjoint as operators}.  As already observed,  the operator $\overline{\pi_\phi(P)}$  is not selfadjoint but it is maximally symmetric so that it admits only one POVM decomposing it and an associated  unique consistent family of measures $\mu^{(P)}_{\phi_B}$ solving the moment problem for all deformations $\phi_B$, $B\in \gB$.  Making use of standard properties of Fourier transform, it is easy to prove that these measures are 
$$\mu^{(P)}_{\phi_B}(E) := \int_{E} \left| \widehat{B\chi}(k)\right|^2 dk\:,\quad E \in \cB(\bR)\:,$$
where
the functions $B\chi$ (for $B\in \gB$) are assumed to be extended to the whole $\bR$ as the zero function for $x\leq 0$ determining Schwartz functions and
 $$\widehat{f}(k) := \frac{1}{\sqrt{2\pi}}\int_\bR e^{-ikx} f(x) dx\:, $$
for $f$ in Schwartz space over $\bR$, is the standard Fourier transform. At this point, it is easy to check that the unique POVM decomposing   $\overline{\pi_\phi(P)}$ is
$$Q^{(\overline{\pi_\phi(P)}, \phi)}(E)  = P_+ P(E)\spa\rest_{L^2([0,+\infty), dx)}\:, \quad E \in \cB(\bR)$$ where $P_+ : L^2(\bR,dx) \to L^2(\bR,dx)$ is the orthogonal projector onto $L^2([0,+\infty),dx)$ viewed as closed subspace of $L^2(\bR,dx)$, and $P$ is the PVM of the standard selfadjoint {\em momentum operator} in $L^2(\bR,dx)$.}\hfill $\blacksquare$
\end{example}
\section{Conclusions and open problems}\label{Sec: Conclusions and open problems}
Before going to the summary we state our overall conclusions:
\begin{itemize}
	\item[(a)]
		When dealing with $*$-algebras the notion of algebraic observable -- $a=a^*\in\mathfrak{A}$ -- and quantum observable -- $\pi_\omega(a)$ (essentially) self-adjoint for a given $\omega$ -- do not necessarily agree as we proved by discussing  some simple counterexamples.
	This raises a problem with  the physical interpretation of those algebraic observables which do not produce  quantum observables in some GNS representation.
	This issue does not arise in the $C^*$-algebraic setting where an algebraic observable always defines quantum observables in every GNS representation.  However, the use of $*$-algebras that are not $C^*$-algebras is in particular mandatory in some important cases as perturbative QFT and the afore-mentioned problem cannot avoided.
	\item[(b)]
	The notion of POVM turned out to be of  pivotal interest in our investigation. On the one hand it provides a universally recognized notion of generalized observable (as is well known  from other areas of quantum physics like quantum information) that can be used in the interpretation of $\pi_\omega(a)$  when it is not essentially selfadjoint  for an algebraic observable $a$.
On the other hand, we also proved that POVMs have a nice interplay  with the moment problem of the pair $(a, \omega)$ and those 
of the deformations $(a, \omega_b)$, $b\in \gA$. The moment problem should  be tackled for  accepting the popular intepretation of
$\omega(a)$ as the  expectation value of the algebraic observable $a$. The moment problem admits a unique (spectral) solution if $\gA$
is a $C^*$-algebra, but as before, it involves some subtle issues when dealing with $*$-algebras.
However, we proved that, for an algebraic observable $a$ in a $*$-algebra $\gA$, if the moment problems of the class of $(a, \omega_b)$, $b\in \gA$, admit unique solutions, then $\pi_\omega(a)$ is essentially selfadjoint (quantum observable) and the POVM of $\pi_\omega(a)$ is a PVM. In the general case, the class of POVMs decomposing a given GNS representation  $\pi_\omega(a)$ of an algebraic observable $a$ are one-to-one with a class of special, physically meaningful,  families of solutions of the moment problems for the pairs $(a, \omega_b)$. There is only one such family if and only if $\pi_\omega(a)$ admits a unique POVM, i.e.,  it is maximally symmetric. From this viewpoint, a maximally symmetric $\pi_\omega(a)$ seems to define  a good generalization of a quantum observable.
\end{itemize}

\subsection{Summary}  {\em Issue A} concerned the fact that an Hermitian element $a^*=a\in \gA$ may be represented in a GNS representation of some non-normalized state $\omega$ by means of an operator $\pi_\omega(a)$ which is not essentially selfadjoint and which can or cannot have selfadjoint extensions (see (1) and (2) in Example \ref{Ex: algebraically self-adjoint observables which are not self-adjoint as operators}).
Stated differently, an algebraic observable does not always define a quantum observables.
If $\gA$ is a $C^*$-algebra, $\pi_\omega(a)$ is always selfadjoint, but there are some cases in physics as QFT where the use of $C^*$-algebras is not technically convenient.
Therefore Issue A must be therefore seriously considered.
We have proved, however, that it is always possible to interpret the symmetric operator $\pi_\omega(a)$ as a 
{\em generalized observable} just by fixing a normalized POVM associated to it which decomposes the operator $\pi_\omega(a)$  according to Definition \ref{DEFassociate} into a generalized version of the spectral theorem of selfadjoint operators.
The POVM decomposing $\pi_\omega(a)$ is unique if and only if  $\overline{\pi_\omega(a)}$ is {\em maximally symmetric} (Definition \ref{Definition: maximally symmetric operator}) 
and this unique normalized POVM is a PVM when  $\pi_\omega(a)$ is essentially selfadjoint.
This provides a sufficient condition for an algebraic observable to uniquely define a generalized observable in a given GNS representation. However, in the general case  there are many POVMs associated to a given symmetric operator  $\pi_\omega(a)$. The reduction of the number of those POVMs is entangled with the next issue.
\\

\noindent {\em Issue B} regarded the popular expectation-value interpretation 
of $\omega(a)$.
Form an  operational point of view in common with the formulation of classical physics, 
 $\mu^{(a)}_\omega$ can be fixed looking for a measure giving rise to the 
known momenta  $\omega(a^n)$, that is solving the moment problem  (\ref{MP}). If $\overline{\pi_\omega(a)}$ 
is selfadjoint, a physically meaningful way (\ref{muPVM}) to define $\mu^{(a)}_\omega$ uses  the PVM of  $\overline{\pi_\omega(a)}$.
 In general, many measures 
$\mu^{(a)}_\omega$ associated with the class of moments $\omega(a^n)$ as in (\ref{MP}) exist
even if   $\overline{\pi_\omega(a)}$ does not admit selfadjoint extensions.
The physical meaning of these measures is dubious. 
The number  of the  measures $\mu_\omega^{(a)}$ is reduced  by considering the information provided by other elements  
$b\in \gA$ in terms of deformed  non-normalized states $\omega_b$  (Definition \ref{defpert}) 
and considering  $\mu_\omega^{(a)}$ as an element of the  class of  measures $\mu_{\omega_{b}}^{(a)}$ 
solving separately the moment problem for $a$ and each $\omega_b$.
These measures are expected to enjoy a list of physically meaningful mutual relations   (\ref{A})-(\ref{B})
 able to considerably reduce their number.
A class of such measures, for $a$ and $\omega$ fixed  is called
  {\em consistent class of measures}  (Definition \ref{Definition: compatibility condition for family of measures}).
 (If $\gA$ is a $C^*$-algebra, there is exactly one measure $\mu_{\omega}^{(a)}$
 spectrally obtained implementing the expectation value interpretation of $\omega(a)$ and Issue B
is harmless.)
 Theorem \ref{theorem-a-POVM}   established that every normalized POVM  decomposing the symmetric operator 
$\pi_\omega(a)$
 defines a unique class of consistent measures $\{\nu^{(a)}_b\}_{b\in\gA}$  in the natural way (\ref{muQ}). These measures   
solve the moment problem for $\omega_b$, thus corroborating the expectation-value interpretation of $\omega_b(a)$ and $\omega(a)$ in particular. 
When the POVM is a PVM, the standard relation  (\ref{muPVM}) between PVMs and Borel spectral measures is recovered.
The result is reversed in Theorem \ref{Theorem: characterization of compatible measures}, which is the main achievement of this paper:  for a Hermitian element $a\in \gA$ and a non-normalized state $\omega$, 
an associated  consistent class of measures $\mu_{\omega_{b}}^{(a)}$ solving the moment problem for every corresponding deformation $\omega_b$ always determines a unique POVM  which decomposes   the symmetric operator $\pi_\omega(a)$.
%
Also the measures 
$\mu_{\omega_b}^{(a)}$ arise  from  POVMs $Q^{(a,\omega_b)}$ (\ref{POVMmu2}) which, as expected,  decompose the respective generalized observables $\pi_{\omega_b}(a)$. The POVMs $Q^{(a,\omega_b)}$ are all induced by the initial POVM $Q^{(a,\omega)}$ (\ref{newQ}).
As a complement,  Corollary \ref{Cor: uniqueness of consistent measures} establishes that $\overline{\pi_\omega(a)}$ is maximally symmetric (selfadjoint in particular) if and only if there is only a unique  class of consistent measures  $\mu_{\omega_{b}}^{(a)}$ solving the moment problem for every corresponding deformation $\omega_b$.
Part of  Corollary \ref{Cor: uniqueness of consistent measures}  admits a stronger version established in  Theorem \ref{T1} which refers to the whole class of measures   $\mu_{\omega_{b}}^{(a)}$ solving the moment problem for every corresponding deformation $\omega_b$  without imposing constrains (\ref{A})-(\ref{B}).  If $a=a^*$ and $\omega$ are fixed and there is exactly one measure solving the moment problem
for each deformation $\omega_b$, then $\overline{\pi_\omega(a)}$ is selfadjoint and therefore has the interpretation of a standard observable in the usual Hilbert space formulation of quantum theories. 
The converse assertion  of this very strong result is untenable, as explicitly proved with two counterexamples from elementary QM and elementary QFT
 (Example \ref{Ex: uniqueness of moment problem and 
self-adjointness of the operator are not related}).

\subsection{Open issues}
There are at least two important open issues after the results established in this work. One concerns the fact that, when $\pi_\omega(a)$ is only symmetric, its interpretation as generalized observable depends on the choice of the normalized POVM associated to it. This POVM is unique if and only if $\overline{\pi_\omega(a)}$ is maximally symmetric (selfadjoint in particular). It is not clear if the information contained in the triple $\gA$, $a$, $\omega$ permits one to fix  this choice 
or somehow  reduce  the number of possibilities. The second open issue regards the option  of simultaneous measurements of compatible (i.e., pairwise commuting) abstract observables $a_1,\ldots, a_n$ with associated joint measures on $\bR^n$ accounting for the expectation-value 
interpretation.
The many-variables moment problem is not a straightforward generalization  of the one-variable moment problem \cite{Schmudgen17} and also the notion of joint POVM presents some non-trivial technical difficulties \cite{Beneduci17}.
Already at the level of selfadjoint observables, commutativity of symmetric operators (say $\pi_\omega(a_1)$ and $\pi_\omega(a_2)$) on a dense invariant domain of essential selfadjointness ($\cD_\omega$) does not imply the much more physically meaningful commutativity of their respective PVMs (as proved by Nelson \cite{RS2}) and the existence of a joint PVM.
These aspects have been investigated in \cite{Borchers-Yngvason-91}, where a stronger positivity condition for the state $\omega$ has been imposed to ensure the existence of self-adjoint extensions for the operators $\pi_\omega(a_1),\ldots,\pi_\omega(a_n)$ whose spectral measures mutually commute.
We plan to investigate this issue in the light of our results in a forthcoming publication.

\section*{Acknowledgments}
The authors thank C. Capoferri, C. Dappiaggi, S. Mazzucchi and N. Pinamonti for useful discussions.
We are grateful to C. Fewster for helpful comments and for pointing out to us reference \cite{Borchers-Yngvason-91} and to the anonymous referees for their  the helpful suggestions.

\appendix
\section{Appendix}\label{Appendix: on POVM}
\subsection{Reducing subspaces}\label{reduction}

\begin{definition}
	{\em  If  $D(T) \subset \mathsf{H}$ is a subspace of the Hilbert space $\sH$, let $T : D(T) \to \mathsf{H}$  be an operator   and 
	$\mathsf{H}_0 \subset \mathsf{H}$  a closed subspace with $\{0\} \neq \mathsf{H}_0 \neq \mathsf{H}$,  $P_0$ denoting 
	the orthogonal projector onto $\mathsf{H}_0$. In this case, $\mathsf{H}_0$ is said to {\bf reduce} $T$ if  both conditions are true
	\begin{itemize}
		\item[(i)] $T(D(T) \cap \mathsf{H}_0) \subset \mathsf{H}_0$ and $T(D(T) \cap \mathsf{H}^\perp_0) \subset \mathsf{H}^\perp_0$,	
		\item[(ii)] $P_0(D(T)) \subset D(T)$,
	\end{itemize}
	so that  the direct orthogonal decomposition holds
	\begin{equation}
		D(T) = (D(T) \cap \mathsf{H}_0) \oplus (D(T) \cap \mathsf{H}^\perp_0)\quad \mbox{and} \quad T= T\spa\rest_{D(T) \cap \mathsf{H}_0}\oplus T\spa\rest_{D(T) \cap \mathsf{H}^\perp_0}\:.\label{dirorthdec}
	\end{equation}
The operator $T_0= T\spa\rest_{D(T) \cap \mathsf{H}_0}$ is called the  {\bf part of $T$ on $ \mathsf{H}_0$}.}\hfill $\blacksquare$\\
\end{definition}

\begin{remark}\label{remred}$\null$\\
	{\em
	{\bf (1)} It is worth stressing that  (i) does {\em not} imply (ii).
	(A counterexample is given by $\mathsf{H}:=L^2([0,+\infty),dx)$, $T:=-i\frac{d}{dx}$ with $D(T):= C_c^\infty((0,+\infty))$ and $\mathsf{H}_0:=\operatorname{span}\{\varphi\}$, where $\varphi(x)=e^{-x}$.
Here $D(T)\cap \mathsf{H}_0 = \{0\}$ so that the former inclusion  in (i) is trivially valid, whereas the latter is valid by integrating by parts. However, it is easy to pick out  $\psi \in D(T)$ such that $\langle \psi|\varphi\rangle \neq 0$, so that 
$(P_0\psi)(x)$  has the support of $\varphi$ itself given by the whole $[0,+\infty)$ and $P_0\psi \not \in D(T)$.)}
	Moreover, without (ii) the direct orthogonal decomposition (\ref{dirorthdec}) cannot take place.  \\
	{\bf (2)} Evidently  $\mathsf{H}_0$ reduces $T$ iff  $\mathsf{H}_0^\perp$ reduces $T$, since  $I-P_0$ projects onto 
	$\mathsf{H}_0^\perp$.\\
	{\bf (3)} A condition equivalent to (i)-(ii) is $P_0T\subset TP_0$ as the reader immediately proves.\hfill $\blacksquare$\\
\end{remark}

\begin{proposition}\label{priopevidente}
	If the closed subspace $\mathsf{H}_0$ reduces  the symmetric (selfadjoint) operator $T$ on $\mathsf{H}$, then also $T\spa\rest_{D(T) \cap \mathsf{H}_0} $ and $T\spa\rest_{D(T) \cap \mathsf{H}^\perp_0}$ are symmetric (resp. selfadjoint) in $\mathsf{H}_0$ and  $\mathsf{H}_0^\perp$ respectively.
\end{proposition}

\begin{proof} Direct inspection.
\end{proof}
\noindent A useful  technical fact is  presented in the following proposition \cite{Schmudgen12}.
\begin{proposition}\label{propSR}
Let $T$ be a closed symmetric operator on a Hilbert space $\sH$. Let $D_0$ be a dense
 subspace of a closed subspace $\sH_0$ of $\sH$ such that $D_0 \subset D(T )$ and $T(D_0)\subset \sH_0$.
Suppose that $T_0 := T\spa\rest_{D_0}$ is essentially
self-adjoint on $\sH_0$. Then $\sH_0$ reduces $T$  and $\overline{T_0}$  is the part of $T$
on $\sH_0$.
\end{proposition}
\begin{proof} See \cite[Prop.1.17]{Schmudgen12}.
\end{proof}

\subsection{More on generalized symmetric and selfadjoint extensions}
\noindent According to definition \ref{Definition: generalized extension of symmetric operator}  we have
	\begin{align}\label{Equation: relation between domain of a symmetric operator and its generalized extension}
		D(A)\subset D(B)\cap\mathsf{H}\subset D(B)\:,
	\end{align}
where $D(B)$ is dense in $\mathsf{K}$ and $D(A)$ is dense in $\mathsf{H}$.
Generalized extensions $B$ with $B\supsetneq A$ are classified accordingly to the previous inclusions following \cite{Akniezer-Glazman-93}, in particular:
	\begin{itemize}
		\item[(i)]
		$B$ is said to be of  {\bf kind I} if $D(A)\neq  D(B)\cap\mathsf{H}= D(B)$ -- that is, if $B$ is a standard extension of $A$;
		\item[(ii)] 
		$B$ is said to be of  {\bf kind II} if $D(A)= D(B)\cap\mathsf{H}\neq D(B)$;
		\item[(iii)]
		$B$ is said to be of  {\bf kind III} if $D(A)\neq  D(B)\cap\mathsf{H}\neq  D(B)$;
	\end{itemize}

\begin{proposition}\label{prop22}
	If $A:D(A) \to \mathsf{H}$ with $D(A) \subset  \mathsf{H}$ is maximally symmetric  and $B \supsetneq A$ is a generalized symmetric extension, then $B$ is of kind $II$.
\end{proposition}

\begin{proof} The kind $I$ is not possible {\em a priori} since $A$ does not admit proper symmetric 
extensions in $\mathsf{H}$. Let us assume that $B$ is either of kind $II$ or $III$ and consider the operator
 $P_{\mathsf{H}}BP_{\mathsf{H}}$, with its natural domain $D(P_{\mathsf{H}}BP_{\mathsf{H}})=
D(BP_{\mathsf{H}})$, where $P_{\mathsf{H}}\in \cL(\mathsf{K})$ is 
the orthogonal projector onto  $\mathsf{H}$. Since this is a symmetric extension of $A$ in $\mathsf{H}$ which is maximally symmetric, we have $A=P_{\mathsf{H}}BP_{\mathsf{H}}$, in particular  $D(BP_{\mathsf{H}})=D(A)$.
As a consequence, if $x \in  D(B)\cap\mathsf{H}$, then $x \in D(BP_{\mathsf{H}}) =D(A)$ so that $D(A) \supset  D(B)\cap\mathsf{H}$ and thus $D(A) =  D(B)\cap\mathsf{H}$ because the other inclusion is true from (\ref{Equation: relation between domain of a symmetric operator and its generalized extension}). We have 
proved that $B$ is of kind $II$.  
\end{proof}
\noindent We have an important technical result 

\begin{theorem}\label{theoremkindII}
	A  non-selfadjoint symmetric operator $A$ always admits a generalized selfadjoint extension $B$. Such an extension can be chosen of  kind $II$ when $A$ is closed.  
\end{theorem}

\begin{proof} 
	See \cite[Thm. 1. p.127 Vol II]{Akniezer-Glazman-93} for the former statement. The latter statement relies on 
the comment under the proof of  \cite[Thm. 1. p.127 Vol II]{Akniezer-Glazman-93} and it is completely proved in
 \cite[Thm.13]{Naimark1940}\footnote{Unfortunately the necessary closedness requirement disappeared passing from  \cite[Thm.13]{Naimark1940} to the 
comment under \cite[Thm. 1. p.127 Vol II]{Akniezer-Glazman-93}.}. The fact that the selfadjoint extensions can be chosen in order to satisfy 
(iii) of Definition \ref{Definition: generalized extension of symmetric operator} is proved in \cite[Thm.7]{Naimark1940b}. 
\end{proof}

\subsection{Proof of some propositions}\label{AppendixB}

\begin{lemma}\label{Laggiunto1}
Referring to (1) in example \ref{Ex: uniqueness of moment problem and self-adjointness of the operator are not related}, $\psi_\omega \in {\cal D}_{\omega_B}$.
\end{lemma}
\begin{proof}
To prove that $\psi_\omega \in {\cal D}_{\omega_B}$ observe  that, since  $A$ and $A^*$ are generators of  
$\mathfrak{A}_{\textrm{CCR},1}$  satisfying $[A,A^*]=I$,  we can always rearrange every element $B \in \mathfrak{A}_{\textrm{CCR},1}$
into the form
$$B= \sum_{n,m }  c^{(B)}_{n,m}A^{*n}A^m\:,$$
where $A^0:=A^{*0}:=I$ and where only a finite number of coefficients $c^{(B)}_{n,m} \in \bC$ depending on $B$ are different from $0$.
 Here, (\ref{GNSb}) implies 
\begin{equation}\psi_{\omega_B} = \sum_{n,m }  c^{(B)}_{n,m}A^{*n}A^m\psi_\omega  = \sum_{k \in \bN}  d^{(B)}_k \psi_k\:,\label{dimaggiunta}\end{equation}
where again only a finite number of coefficients $d^{(B)}_k \in \bC$ depending on $B$ are different from $0$
and to pass from the first sum to the second one  we took advantage of  the standard harmonic oscillator algebra of the Hermite basis where  $A\psi_k=\sqrt{k}\psi_{k-1}$,  $A^*\psi_k=\sqrt{k+1}\psi_{k+1}$ with  $\psi_0:= \psi_\omega$.
To conclude, notice that, if $C \in \gA_{\textrm{CCR},1}$, (\ref{GNSb})  also implies that $\pi_{\omega_B}(C)\psi_{\omega_B}= \pi_{\omega}(C)\psi_{\omega_B}$. Therefore, if  $k_B\in \bN$ is the largest natural such that $d^{(B)}_k \neq 0$ in (\ref{dimaggiunta}), choosing  $C:= \frac{1}{d^{(B)}_{k_B}\sqrt{k_B!}}A^{k_B}\in \gA_{\textrm{CCR},1}$, we have
$$\pi_{\omega_B}(C)\psi_{\omega_B} =
 \sum_{k \in \bN}  d^{(B)}_k \frac{1}{d^{(B)}_{k_B}\sqrt{k_B!}} A^{k_B}\psi_k = 0+ 
 \frac{d^{(B)}_{k_B} \sqrt{k_B!}}{d^{(B)}_{k_B} \sqrt{k_B!}}\psi_0= \psi_0 = \psi_\omega\:.$$
In other words, $\psi_\omega \in {\cal D}_{\omega_B}$ as argued.
\end{proof}
\begin{lemma}\label{Laggiunto2} Referring to (2) in example \ref{Ex: uniqueness of moment problem and self-adjointness of the operator are not related}, 
$\psi_\omega \in \cD_{\omega_b}$ if (\ref{iK}) is valid.
\end{lemma}
\begin{proof}
First  define the  elements of the algebra $\gA[M,g]$ 
$$A_{\varphi} :=\frac{1}{2} \left(\Phi[\varphi] +  i \Phi[\varphi'] \right)\quad \mbox{and}\quad A^*_{\varphi} :=\frac{1}{2} \left(\Phi[\varphi] -  i \Phi[\varphi'] \right)\:.$$
We stress that the elements $A^*_{\varphi} $ and $A_{\varphi} $ form  a set of generators of $\gA[M,g]$, because the
$\Phi[\varphi]$ are generators and the  previous relations can be inverted to
$$\Phi[\varphi] = A_{\varphi} + A^*_{\varphi}$$
By using (\ref{PHI}), linearity of $x \mapsto a_x^+$, anti-linearity of $x \mapsto a_x$, and (\ref{iK}), we immediately find
\beq \pi_\omega\left( A_{\varphi}\right) = a_{K\varphi}|_{\cD_\omega} \quad \mbox{and} \quad\pi_\omega\left( A^*_{\varphi}\right) = a^+_{K\varphi}|_{\cD_\omega}\label{AAA2}\eeq
The  algebraic relations reflecting (\ref{algQFT}) are also valid from (b) and (c) in item (2) of example \ref{Ex: uniqueness of moment problem and self-adjointness of the operator are not related},
$$[A_{\varphi} , A^*_{\tilde{\varphi}} ] = \langle K\varphi, K\tilde{\varphi}\rangle I\:, \quad 
[A_{\varphi} , A_{\tilde{\varphi}} ]  = [A^*_{\varphi} , A^*_{\tilde{\varphi}} ] =0\:.$$
Exploiting the found results, we can prove that $\psi_\omega \in \cD_{\omega_b}$ as wanted.
 The generic  element $b \in \gA[M,g]$,
 taking advantage of the generators $A_{\varphi}$ and $A^*_{\varphi}$ and their commutation relations, is always of the form
(where we adopt the  convention that $\prod_{j=1}^{0} A^{(*)}_{\varphi_j}:= \bI$)
\beq b =
\sum_{k,h=0}^{N} c_{kh}^{j_1\ldots j_k i_1\ldots i_h} A^*_{\varphi^{(k)}_{j_1}}\cdots  A^*_{\varphi^{(k)}_{j_k}}   A_{\tilde{\varphi}^{(h)}_{i_1}}\cdots   A_{\tilde{\varphi}^{(h)}_{i_h}}\label{bgen}\eeq
 for sets of solutions 
 $\{\varphi^{(k)}_{j}\}_{j\in \{1,\ldots, D_k\}}$ and  $\{\tilde{\varphi}^{(h)}_{i}\}_{i \in \{1,\ldots, E_h\}}$ 
 in $Sol[M,g]$ and
where we adopted Einstein's summation convention. Obviously, $N,D_h, E_h$, and the complex  coefficients $c_{k,h}^{j_1\ldots j_k i_1\ldots i_h} $ depend on $b$ and these coefficients are completely symmetric 
separately in the indices $j_r$ and in the indices $i_r$ due to the fact that the elements  $A^*_{\varphi_{j_1}},\cdots,  A^*_{\varphi_{j_k}}$ and 
$A_{\varphi_{i_1}},\cdots,A_{\varphi_{i_h}}$ separately pairwise  commute. 
  Hence, (\ref{AAA2}) entails 
\beq \psi_{\omega_b} = \pi_{\omega}(b) \psi_\omega =
\sum_{k=0}^{N} c_{k0}^{j_1\ldots j_k} a^+_{\psi^{(k)}_{j_1}}\cdots  a^+_{\psi^{(k)}_{j_k}}\psi_\omega\label{bgen2}\eeq
where
$\psi^{(k)}_j := K\varphi^{(k)}_j$ and we henceforth assume that
 not all coefficients $c_{N0}^{j_1\ldots j_N}$ vanish and all  vectors $\psi^{(N)}_j = K\varphi^{(N)}_j$ for $j=1,\ldots, D_N$ do not vanish 
(otherwise the $N$-th addend  in (\ref{bgen2})  would give  no contribution).
Defining
$$d :=    \overline{c_{N0}^{i_1\cdots i_N}} A_{K\varphi^{(N)}_{i_1}} \cdots A_{K\varphi^{(N)}_{i_N}} \in \gA[M,g]$$
(\ref{algQFT}) yields, if 
$c^{j_1\cdots j_N} := c_{N0}^{j_1\cdots j_N}$ and $\psi_j := \psi^{(N)}_j$
\beq\cD_{\omega_b} \ni \pi_{\omega}(d) \pi_\omega(b) \psi_\omega  =\overline{c^{i_1\cdots i_N}}c^{j_1\cdots j_N} a_{\psi_{i_1}} \cdots  a_{\psi_{i_N}}
 a^+_{\psi_{j_1}}\cdots  a^+_{\psi_{j_N}}
   \psi_\omega  \label{Ldc}\:.\eeq
At this juncture we can fix an orthonormal basis $\{e_l\}_{i=1,\ldots, D}$ in the span of the vectors $\{\psi_{j}\}_{j=1,\ldots, D_N}$ and we can rearrange the identity above as
 $$\pi_{\omega}(d) \pi_\omega(b) \psi_\omega  =\overline{c'^{i_1\cdots i_N}}c'^{j_1\cdots j_N} a_{e_{i_1}} \cdots  a_{e_{i_N}}
 a^+_{e_{j_1}}\cdots  a^+_{e_{j_N}}
   \psi_\omega $$
where not all the symmetric coefficients $c'^{j_1\cdots j_N}$ vanish (since they are the components in a new basis  of the non-vanishing tensor defined by
 the components $c^{j_1\cdots j_N}$). Expanding the contractions we find, where $P_N$ denotes the group of permutations of $N$ objects
$$
   \pi_{\omega}(d) \pi_\omega(b) \psi_\omega = \sum_{\sigma \in P_N} \overline{c'^{i_1\cdots i_N}}c'^{j_1\cdots j_N} \delta_{i_1j_{\sigma(1)}}\cdots \delta_{i_N j_{\sigma(N)}}\psi_\omega
$$ $$ =\sum_{\sigma \in P_N} \overline{c'^{i_1\cdots i_N}}c'^{j_{\sigma^{-1}(1)}\cdots j_{\sigma^{-1}(N)}} \delta_{i_1j_1}\cdots \delta_{i_N j_N}\psi_\omega
=\sum_{\sigma \in P_N} \overline{c'^{i_1\cdots i_N}}c'^{j_1\cdots j_N} \delta_{i_1j_1}\cdots \delta_{i_N j_N}\psi_\omega\:.
$$
(In the second line, $c'^{j_{\sigma^{-1}(1)}\cdots j_{\sigma^{-1}(N)}}= c'^{j_1\cdots j_N} $ arises from  the symmetry of the coefficients.)
That is
\beq \pi_{\omega}(d) \pi_\omega(b) \psi_\omega  = \left(N! \sum_{j_1, \ldots, j_N =1}^{D} |c'^{j_1\cdots c_N}|^2\right) \psi_\omega \in \cD_{\omega_b}\:.\label{Fpsi}\eeq
Since  the coefficient in front of $\psi_\omega$ in (\ref{Fpsi}) does not vanish,  then  $\psi_\omega \in \cD_{\omega_b}$ as wanted.
\end{proof}

\noindent {\bf Proof of Proposition \ref{prop1}}.
Let $B$ be a generalized symmetric extension in the Hilbert space $\sK$ of the selfadjoint operator $A$ in the Hilbert space $\sH$ with $\sH\subset \sK$.
The closed symmetric operator $B'= \overline{B}$ in $\mathsf{K}$ extends $A$.
In view of Proposition \ref{priopevidente}-\ref{propSR}, since
$B'\spa\rest_{D(A)} = A$
is selfadjoint on $\mathsf{H}$, then
$\mathsf{H}$ reduces $B'$ and
$\overline{B} = \overline{B'\spa\rest_{D(A)}} \oplus  B''$,
where the two addends are symmetric operators respectively on 
$\mathsf{H}$ and $\mathsf{H}^\perp$.
Requirement (iii) in Definition \ref{Definition: generalized extension of symmetric operator} imposes that $\mathsf{H}^\perp = \{0\}$, so that $\mathsf{K}= \mathsf{H}$ and $B$ is a 
standard symmetric extension of the selfadjoint operator $A$ which entails $B=A$. \hfill $\Box$\\

\noindent {\bf Proof of Theorem \ref{Theorem: uniqueness of POVM associated with maximally symmetric operator}}.
(a) and (b).  If $A$ is symmetric,  for each generalized selfadjoint extension $B$ as in Definition \ref{Definition: generalized extension of symmetric operator} (they exist in view of Theorem \ref{theoremkindII} and, if $A$ is selfadjoint, $B:=A$), one can define a corresponding 
normalized POVM $Q^{(A)}\colon\mathscr{B}(\mathbb{R})\to\mathfrak{B}(\mathsf{H})$ as follows.
	Let $P\colon\mathscr{B}(\mathbb{R})\to\cL(\mathsf{K})$ be the unique PVM associated with $B$ 
\begin{equation}\label{defB} B := \int_{\mathbb{R}}\lambda dP(\lambda)\,,\end{equation}
(see, e.g. \cite[Thm. 9.13]{Moretti2017}) and let $P_{\mathsf{H}}\in \cL(\mathsf{K})$ denote the orthogonal projection onto
 $\mathsf{H}$ viewed as a closed subspace of $\mathsf{K}$.
	A normalized POVM $Q^{(A)}\colon\mathscr{B}(\mathbb{R})\to\mathfrak{B}(\mathsf{H})$ is then defined by setting
	\begin{align}\label{Equation: POVM associated with generalized extension of symmetric operator}
		Q^{(A)}(E):=P_{\mathsf{H}}P(E)\spa\rest_{\mathsf{H}}\qquad\forall E\in\mathscr{B}(\mathbb{R})\,.
	\end{align}
	The POVM $Q^{(A)}$ is linked to $A$ by the following identities as the reader easily proves from standard spectral theory:
	\begin{align}\label{Equation: relations between POVM and symmetric operator}
		\langle\psi|A\varphi\rangle=\int_{\mathbb{R}}\lambda dQ^{(A)}_{\psi,\varphi}(\lambda)\:, \qquad
		\|A\varphi\|^2=\int_{\mathbb{R}}\lambda^2dQ^{(A)}_{\varphi,\varphi}(\lambda)\,,\qquad\forall\psi\in\mathsf{H},\varphi\in D(A)\,,
	\end{align}
	The above relation implies the following facts, where $\lambda^k$ henceforth denotes the map 
${\mathbb R} \ni \lambda
 \to \lambda^k \in 
{\mathbb R} $ for $k\in \mathbb{N} := \{0,1,2, \ldots\}$,
	\begin{align}\label{Equation: relation between domain of symmetric operator and $L^2$ function w.r.t POVM}
	A= B\spa\rest_{D(A)}\:,\quad
	D(A)\subset\{\psi\in\mathsf{H}|\; \lambda \in L^2(\mathbb{R},Q^{(A)}_{\psi,\psi})\}&=
	\{\psi\in\mathsf{H}|\;\lambda\in L^2(\mathbb{R},P_{\psi,\psi})\}\nonumber \\
	&= D(B)\cap\mathsf{H}\subset D(B)\,,
	\end{align}
the second equality arising from the identity
	\begin{align*}
		\int_{\mathbb{R}}\lambda^2dP_{\psi,\psi}(\lambda)=
		\int_{\mathbb{R}}\lambda^2dQ^{(A)}_{\psi,\psi}(\lambda)\,\qquad\forall\psi\in\mathsf{H}\,,
	\end{align*}
	which descends from equation \eqref{Equation: POVM associated with generalized extension of symmetric operator}.
	\cite[Thm. 2, p. 129 Vol II]{Akniezer-Glazman-93} proves that every POVM (it is equivalent to a {\em spectral function} used therein, see Remark \ref{remequ}) satisfying (\ref{Equation: relations between POVM and symmetric operator}) arises from a PVM of a generalized selfadjoint extension $B$ of $A$
as above, so that (\ref{Equation: relation between domain of symmetric operator and $L^2$ function w.r.t POVM}) are satisfied. The proof of (a) and (b) is over.\\
\noindent (c) If the selfadjoint operator $B$ extending $A$ as in (b) is a standard extension of $A$, then $\sH=\sK$ so that $P_{\sK}=I$ and $Q^{(A)}=P$ so that $Q^{(A)}$ is a PVM.
If $Q^{(A)}$ is a PVM, $B := \int_\bR \lambda dQ^{(A)}(\lambda)$ with domain $D(B)= \{\psi \in \sH \:|\: \int_\bR \lambda^2 dQ^{(A)}_{\psi,\psi}(\lambda) <+\infty\}$  is a standard 
 selfadjoint extension of $A$ and the associated trivial dilation triple $\sK:=\sH$, $P_\sH:=I$, $P:=Q^{(A)}$ trivially generates $Q^{(A)}$ as in (b).\\
\noindent (d) Referring to the proof of (a) and (b) above, 	for generalized selfadjoint extensions $B$ of  kind $II$, and  this choice for $B$ is always feasible when $A$ is closed in view of  Theorem \ref{theoremkindII},  we have $D(A)=D(B)\cap\mathsf{H}$ and therefore
\begin{align}\label{Equation: relation 2 between domain of symmetric operator and $L^2$ function w.r.t POVM}
	A= B\spa\rest_{D(B) \cap \mathsf{H}}\:,\quad
	D(A)= \{\psi\in\mathsf{H}|\; \lambda \in L^2(\mathbb{R},Q^{(A)}_{\psi,\psi})\}
	&=\{\psi\in\mathsf{H}|\;\lambda\in L^2(\mathbb{R},P_{\psi,\psi})\}\nonumber \\
	&= D(B)\cap\mathsf{H}\:.
\end{align}
 Since (\ref{Equation: relation between domain of symmetric operator and $L^2$ function w.r.t POVM}) 
is valid for every POVM satisfying (\ref{Equation: relations between POVM and symmetric operator}), we have that  the identity 
$D(A)=\{\psi\in\mathsf{H}|\;\lambda\in L^2(\mathbb{R},P_{\psi,\psi})\}$ holds if and only if $A= B\spa\rest_{D(B) \cap \mathsf{H}}$, namely
$B$ is an extension of type $II$ of $A$.\\
(e) and (f). They are established in \cite[Thm. 2, p. 135 Vol II]{Akniezer-Glazman-93}, taking Theorem \ref{theoremkindII} into account and (d) above.
\hfill $\Box$\\

\noindent {\bf Proof of Theorem \ref{TheoremPOVMtoA}}.
	Consider Naimark's dilation triple of $Q$,  $(\mathsf{K}, P_{\mathsf{H}}, P)$ and define the selfadjoint operator 
	$B := \int_{\mathbb{R}}\lambda dP(\lambda)$.
	By hypothesis $D(A^{(Q)})= D(B)\cap \mathsf{H}$.
	Since $D(B)$ is a subspace of $\mathsf{K}$, it also holds that  $D(A^{(Q)})$ is a subspace of $\mathsf{H}$ proving (a).
	A well-known counterexample due to Naimark \cite{Akniezer-Glazman-93} proves that, in some cases, $D(A^{(Q)}) =\{0\}$ though $Q$ is not trivial.
	It is  clear that, if an operator $A^{(Q)}: D(A^{(Q)}) \to \mathsf{H}$ satisfies (\ref{eqfixA}) then it is unique due to the arbitrariness of $\varphi \in \mathsf{H}$.
	So we prove (b) and (d) simultaneously just checking that $P_{\mathsf{H}} \int_{\mathbb{R}} \lambda dP(\lambda)\spa\rest_{D(A^{(Q)})}$ satisfies  (\ref{eqfixA}).
	The proof is trivial since, from standard properties of the integral of a PVM, if $\psi \in D(A^{(Q)})\subset D(B)$ and $\varphi \in \mathsf{H}$, then
	\begin{align*}
		\left\langle \varphi \left| P_{\mathsf{H}} \int_{\mathbb{R}} \lambda dP(\lambda)\psi \right. \right\rangle=
		\left\langle  P_{\mathsf{H}}\varphi \left| \int_{\mathbb{R}} \lambda dP(\lambda)\psi \right. \right\rangle=
		\int_{\mathbb{R}}\lambda dP_{P_{\mathsf{H}}\varphi, \psi} (\lambda)=
		\int_{\mathbb{R}}\lambda dQ_{\varphi, \psi}(\lambda)\,.
	\end{align*}
	(c)  immediately arises with the same argument taking advantage of the fact that $B=B^*$ and  $P_{\mathsf{H}} \psi = \psi$ if $\psi \in \mathsf{H}$.
	Regarding (e), it is sufficient observing that  (see, e.g. \cite{Moretti2017}), if $\psi \in D(B)$, then $||B\psi||^2 = \int_{\mathbb{R}} \lambda^2 dP_{\psi,\psi}(\lambda)$ and next taking advantage of $D(A^{(Q)})= D(B)\cap \mathsf{H}$ and  (d) observing  in particular that $P_{\mathsf{H}}^*  P_{\mathsf{H}}\varphi =  P_{\mathsf{H}}  P_{\mathsf{H}}\varphi = P_{\mathsf{H}}\varphi =\varphi$ if $\varphi \in \mathsf{H}$. The fact that $A^{(Q)}$ is closed immediately follows from the fact that $B$ is closed (because selfadjoint) and $A= B|_{D(B) \cap \mathsf{H}}$ where 
$\mathsf{H}$ is closed. \hfill $\Box$\\
 
\noindent {\bf Proof of Proposition \ref{technicalproposition}}.
What we have to prove is just  that the right-hand side of  (\ref{scalarprodfromp}) is a positive semi definite  Hermitian scalar product over $X$.
Indeed, with that definition of the scalar 
product $(\: |\: )_p$, the identity   $p(x) = \sqrt{(x|x)_p}$ is valid (see below) and this fact automatically  implies that   $p$ is a seminorm. Uniqueness of the scalar product generating a seminorm $p$ 
is a trivial consequence of the polarization identity. Let us prove that $(\: |\: )_p$ defined in  (\ref{scalarprodfromp})
is a positive semidefinite  Hermitian 
scalar product  over $X$. From the definition of $(\:|\:)_p$ and property (i)  of $p\colon X\to [0,+\infty)$, it is trivial business  to prove the following facts per direct inspection:
\begin{itemize}
\item[1.] $(x|x)_p = p(x)^2 \geq 0$,
\item[2.] $(x|0)_p=0$,
\item[3.] $(x|y)_p= \overline{(y|x)_p}$,
\item[4.] $(x|iy)_p = i(x|y)_p$.
\end{itemize}
With these identities, we can also prove 
\begin{itemize}
\item[5.] $(x|y+z)_p = (x|y)_p + (x|z)_p$,
\item[6.] $(x|-y)_p = -(x|y)_p$,
\end{itemize}
by exploiting property (ii) of $p$. Actually,  (6) immediately arises form (2) and (5). We will prove (5) as the last step of this proof.
Iterating property (5), we easily obtain $(x|ny)_p = n(x|y)_p $ for every $n\in\mathbb{N}$, so that $(1/n) (x|z)_p = (x|(1/n) z)_p $
 when replacing $ny$ for $z$.  As a consequence, $\lambda (x|y)_p = (x|\lambda y)_p$ if $\lambda \in \mathbb{Q}$. This results extends to $\lambda \in \mathbb{R}$ if the map $\mathbb{R} \ni \lambda \mapsto  (x|\lambda y)_p$ is right-continuous because, for every $x\in [0,+\infty)$ there is a decreasing sequence of rationals tending to it. The definition 
 (\ref{scalarprodfromp})  of  $(x|y)_p$ proves that map is in fact right-continuous since property (iii) of $p$ implies that, if $\lambda_0 \in [0, +\infty)$,
$$p(x + \lambda i^k y) = p( (x+ \lambda_0 i^ky) + (\lambda-\lambda_0) i^k y) \to  p(x+ \lambda_0 i^ky) \quad \mbox{for $\mathbb{R}\ni \lambda \to \lambda_0^+$}\:.$$
We can therefore add the further property
\begin{itemize}
\item[7.] $\lambda (x|y)_p = (x|\lambda y)_p$ if $\lambda \in [0,+\infty)$.
\end{itemize}
Collecting properties (1), (3), (5), (7), (6), (4) together,  we obtain that  $X\times X \ni (x,y) \mapsto (x|y)_p$ defined as in (\ref{scalarprodfromp}) is a positive semi definite Hermitian scalar product over $X$ whose associated seminorm is $p$ as wanted.\\
To conclude the proof, we establish property (5) from requirement  (ii) on $p$.
$$4(x|y+z)_p = \sum_{k=0}^3(-i)^k p(x+i^k(y + z))^2=  \sum_{k=0}^3(-i)^k p((x/2+i^ky) + (x/2 + i^k z))^2\:.$$
Since, from (i) of the requirements on $p$,
$$\sum_{k=0}^3(-i)^k p((x/2+i^ky) - (x/2 + i^k z))^2 = \sum_{k=0}^3(-i)^k p( i^k (y-z))^2 = \sum_{k=0}^3(-i)^k p( (y-z))^2 =0,$$
we can re-arrange the found decomposition of $4(x|y+z)_p$ as
$$4(x|y+z)_p =   \sum_{k=0}^3(-i)^k \left[ p((x/2+i^ky) + (x/2 + i^k z))^2 -  p((x/2+i^ky) - (x/2 + i^k z))^2\right]\:.$$
 From (ii) of the requirements on $p$,
$$4(x|y+z)_p =   \sum_{k=0}^3(-i)^k  \left[ 2p(x/2+i^ky)^2 +  2p(x/2+i^kz)^2\right] = 8(x/2|y)_p + 8(x/2|z)_p\:.$$
The special case $z=0$ and (2) yield  $2(x/2|y)_p= (x|y)_p$ which, exploited again above, yields the wanted result  (5) 
$4(x|y+z)_p  = 4(x|y)_p + 4(x|z)_p$. \hfill  $\Box$

\end{document}